\documentclass[
 twocolumn,prl,
 amsmath,amssymb,
 aps
]{revtex4-2}

\usepackage{graphicx}% Include figure files
\usepackage{hyperref}
\hypersetup{pdfborder={0 0 0},colorlinks,breaklinks=true,
 urlcolor={purple},citecolor={purple},linkcolor={purple} }

\usepackage{dcolumn}% Align table columns on decimal point
\usepackage{bm}% bold math
\usepackage{mathtools}
\usepackage{physics}
\usepackage{bbm}
\usepackage{mathrsfs}
\usepackage{amsthm}
\usepackage{xcolor}
\usepackage{empheq}
\usepackage{cases}
\usepackage{float}

\newcommand{\red}[1]{\textcolor{black}{#1}}
\def\Z{\mathbb{Z}}

\def\C{\mathbb{C}}

\def\cA{\mathcal{A}}
\def\cB{\mathcal{B}}

\def\cD{\mathcal{D}}
\def\cN{\mathcal{N}}
\def\cM{\mathcal{M}}
\def\cI{\mathcal{I}}
\def\cS{\mathcal{S}}

\def\se{\subseteq}
\def\ss{\subset}

\newtheoremstyle{qca}
{3pt}% ?Space above ?1 
{3pt}% ?Space below ?1
{\itshape}% ?Body font ?
{3pt}% ?Indent amount ?2
{\itshape}% ?Theorem head font?
{.}% ?Punctuation after theorem head ?
{3pt}% ?Space after theorem head ?3
{}%
\theoremstyle{qca}
\newtheorem{theorem}{Theorem}

\newtheorem{lemma}{Lemma}

\newtheorem{definition}{Definition}

\begin{document}

%\preprint{APS/123-QED}

%-----------------------------------------------------------------
%----------------------------- Preamble
%-----------------------------------------------------------------
\title{Classification of qubit cellular automata on hypercubic lattices}

\author{Andrea Pizzamiglio }
\email{andrea.pizzamiglio01@universitadipavia.it}
\affiliation{QUIT group, Dipartimento di Fisica, Universit\`a degli studi di Pavia, and INFN sezione di Pavia, via A. Bassi 6, 27100 Pavia, Italy}
\author{Alessandro Bisio  }
\email{alessandro.bisio@unipv.it}
\affiliation{QUIT group, Dipartimento di Fisica, Universit\`a degli studi di Pavia, and INFN sezione di Pavia, via A. Bassi 6, 27100 Pavia, Italy}
\author{Paolo Perinotti }%
 \email{paolo.perinotti@unipv.it}
\affiliation{QUIT group, Dipartimento di Fisica, Universit\`a degli studi di Pavia, and INFN sezione di Pavia, via A. Bassi 6, 27100 Pavia, Italy}

%\collaboration{Collaboration}%\noaffiliation


\begin{abstract}
We classify \red{Quantum Cellular Automata whose cells are qubits}, on hypercubic lattices $\mathbb Z^s$, with the von Neumann neighborhood scheme, in terms of \red{realizability} as finite-depth quantum circuits. We show the most general structure of such \red{automata} and use its characterisation to simulate a few steps of evolution and evaluate the rate of entanglement production between one cell and its surroundings.
\end{abstract}

\maketitle

%\tableofcontents

%-----------------------------------------------------------------
%----------------------------- Contents
%-----------------------------------------------------------------
Quantum Cellular Automata (QCA) represent the most general time-discrete, unitary and local dynamics of a lattice of quantum systems~\cite{schumacher2004reversible,Farrelly2020reviewofquantum,arrighi2019overview}. Local means that the state of a cell at time $t+1$ depends only on a finite neighborhood of the cell at time $t$, ensuring a bounded information propagation speed across the lattice. In the present Letter, we classify translation-invariant nearest-neighbor QCA of $s$-dimensional hypercubic lattices $\mathbb{Z}^s$ \red{with qubit cells}. The classification includes a \red{general realization scheme for those QCA corresponding} to Finite-Depth Quantum Circuits (FDQCs). QCA epitomize the simplest scenario for many-body physics toy models and effective and reliable forthcoming quantum devices. Moreover, the classification provides us with a sandbox for developing techniques such as renormalization~\cite{Trezzini2025renormalisationof,rotundo2024effectivedynamicsminimisingdissipation}, discrete-time perturbation theory~\cite{PhysRevLett.126.250503}, statistical mechanics without Hamiltonians, that will provide a toolkit for discrete space-time quantum field theory models~\cite{BISIO2015244,PhysRevA.88.032301,PhysRevResearch.6.033136} beyond 1+1 dimensions~\cite{PhysRevA.90.062106,DAriano:2016aa,PhysRevA.96.062101,Perinotti_2020,Arrighi_2014,Eon2023relativistic}. \red{Specializing to nearest-neighbor qubit QCA on $\mathbb Z^s$, the limited subalgebras structure of single-cell operator algebra makes the analysis of algebraic conditions for the update rules of a QCA amenable to characterization by exhaustion.} The case $s=1$ was treated in Ref.~\cite{schumacher2004reversible}. The above classification fully characterizes the subclass of Clifford QCA, mapping Pauli matrices to tensor products of Pauli matrices \cite{haah2021clifford,10.1063/1.3278513,haah2022topological}, thus leaving stabilizer states invariant~\cite{PhysRevA.101.062302}. 

Further, we reconsider the classification in terms of local flow of information. Information behaves like an uncompressible fluid, with a flow constant along the chain of cells in dimension $1+1$. This flow is quantified by the \emph{index} \cite{GNVW}, that completely specifies the local invariants of a QCA in dimension $1$.  
Two QCA have the same index if and only if they can be transformed into each other by concatenation with an FDQC~\cite{GNVW}. In particular, QCA reducing to FDQCs are precisely those with index $1$. 
Although it is natural to expect that this result generalizes to higher spatial dimensions, naive extensions fail \cite{Vogts2009DiscreteTQ, freedman2020classification, haah2023nontrivial,haah2021clifford}. 

%However, classification of QCA in terms of an index can be simply extended to those qubit QCA that we explicitly classified.  
\red{In dimension $s$, indices for the flow of information along all Cartesian directions can be defined for nearest-neighbor qubit  QCA, and fully classify the latter modulo FDQCs.  For qudits, conditions that allow for classification cannot be derived in general, thus we enforce them as additional constraints on the local update rule, generalizing the classification to QCA with cells of prime dimension.}

This result marks an instance of an index in dimension~$\geq 3$, and although it pertains a minimal system, it retains considerable interest for applications in quantum computing and the topological phases of matter ~\cite{stephen2019subsystem,po2016chiral,fidkowski2019interacting,kitaev2006topological,haah2022topological}, quantum chaos and scrambling dynamics~\cite{gong2021topological,ranard2022converse}. 

To substantiate the significance of our classification, we leveraged it to simulate various families of QCA and studied their entangling power versus the parameters of their quantum gates.
% showcasing the potential wealth of applications of our classification. 
Quantifying entanglement~\cite{horodecki2009quantum} at each stage of a process represents indeed a significant indication of its complexity.

A classical cellular automaton, such as Conway's game of life, comprises a lattice of cells, each containing a $d$-level system, along with a \emph{homogeneous} and \emph{local} update rule, i.e.~the same function updates the state of any cell $\bm x$ at step $t+1$ according to the state of neighboring cells of $\bm x$ at step $t$.
% is a lattice of cells, each one containing a $d$-level system, along with an update rule that is \emph{local}, i.e.~it changes the state of any cell $\bm x$ from step $t$ to step $t+1$, evaluating it as a function of the state of neighbouring cells of $\bm x$ at step $t$, and \emph{homogeneous}, i.e.~the function is the same for every cell. 
However, in the quantum case, locality cannot be easily defined due to the no-cloning theorem~\cite{Wootters:1982aa}. Specifically, evaluation of a quantum ``function''
% ---the local evolution rule---
on the neighborhood of $\bm x$, necessarily alters its state, interfering with the evaluation of the state of any cell $\bm y$ whose neighborhood intersects that of $\bm x$. For this reason, in the quantum scenario it is much easier to define the local update avoiding reference to states, i.e.~recurring to the Heisenberg picture. 

In this perspective, 
% locality amounts to the requirement that an operator that is localised in a cell $\bm x$---intuitively speaking, an operator of the form $O_{\bm x}\otimes I$, with \red{identity} $I$ acting on all the remaining cells---after one step of evolution represented by some completely positive map $\alpha$, gets spread only over the neighbourhood $\cN_{\bm x}$ of $\bm x$, i.e. $\alpha(O_{\bm x}\otimes I)=O'_{\cN_{\bm x}}\otimes I$.  The subscripts indicate the sites where a given \red{operator} (algebra) is supported. Homogeneity, on the other hand, is phrased as commutation of $\alpha$ with translations on the lattice, that are in turn defined in the Heisenberg picture as $\tau^{\bm x}(O_{\bm y}\otimes I)=O_{\bm x+\bm y}\otimes I$ \footnote{We used the additive notation which is suitable only if translations on the lattice are abelian. This is the case for the remainder of the paper, even though more general situations can be conceived~\cite{ARRIGHI_MARTIEL_NESME_2018,DAriano:2016aa,Perinotti_2020}}.
% We now review the above notion of a quantum cellular automaton on $\mathbb Z^s$, in a more 
% rigorous way. 
with each site $\bm{x} \in \mathbb{Z}^s$ we associate the C*-algebra $\mathcal{A}_{\bm{x}}$ of bounded operators on the corresponding Hilbert space~\footnote{We remind the reader that a C*-algebra $\mathcal A$ is an algebra of operators over some complex Hilbert space that is closed under adjoint and complete (every Cauchy sequence in the operator norm converges within the algebra itself), with $\|A^\dag A\|=\|A\|^2$ and $\|AB\|\leq\|A\|\,\|B\|$ for all elements $A,B\in\mathcal A$.}, which is isomorphic to the algebra of complex $d \times d$ matrices $\mathcal{A}_{\bm{x}} \cong \cM_d$.
Let $\left[ \Z^s\right]$ denote the set of \emph{finite} subsets of $\Z^s$. 
The algebra corresponding to a region $\Lambda \in \left[ \Z^s\right]$ with $n$ sites is $\mathcal{A}_{\Lambda} \coloneqq \bigotimes_{\bm{x} \in \Lambda} \mathcal{A}_{\bm{x}} \cong \mathcal{M}_{d^n}$. We denote the trivial subalgebra of $\mathcal A_X$ by $\cI_X \coloneqq \C\, I_{X}$.  For $\Lambda_1 \subset \Lambda_2 \in \left[ \Z^s\right]$, the algebra $\mathcal{A}_{\Lambda_1}$ is isomorphic to the subalgebra $\mathcal{A}_{\Lambda_1} \otimes \cI_{\Lambda_2 \setminus \Lambda_1}$ of $\mathcal{A}_{\Lambda_2}$. The product $O_{\Lambda_1} O_{\Lambda_2}$ of operators $O_{\Lambda_i} \in \cA_{\Lambda_i}$ is thereby a well-defined element of $\cA_{\Lambda_1 \cup \Lambda_2}$. The above construction provides the {\em local algebra}, which is the \emph{inductive limit} of the union of algebras of finite regions~\cite{zbMATH01216133}.
The local algebra can be completed in the uniform operator norm $\|O\|_\infty$ yielding a C*-algebra $\mathcal A_{\mathbb Z^s}$ called \emph{quasi-local} algebra on $\mathbb Z^s$, thus encompassing all operators that are arbitrarily well approximated by sequences of local operators~\cite{doi:10.1142/4090,brattelirobinson}. The system evolves in discrete time steps, by reiterating the QCA, whose definition follows.

\begin{definition}\cite{schumacher2004reversible} A \emph{Quantum Cellular Automaton} on $\mathbb Z^s$ with neighborhood scheme $\mathcal N\in[\mathbb Z^s]$ is an automorphism of C*-algebras $\alpha:\mathcal A_{\mathbb Z^s}\to\mathcal A_{\mathbb Z^s}$ such that
\begin{enumerate}
\item
for $O_{\bm x}\in\mathcal A_{\bm x}$, $\alpha(O_{\bm x})\in\mathcal A_{{\bm x}+\mathcal N}$;
\item $\alpha$ commutes with translations, i.e.~for every $\bm y\in\mathbb Z^s$,
$\tau^{\bm y}\,\alpha=\alpha\,\tau^{\bm y}$.
\end{enumerate}
\end{definition}
\noindent
Note that $\tau^{\bm{x}}: \mathcal{A}_{\mathbb{Z}^s} \to \mathcal{A}_{\mathbb{Z}^s}$ is itself a QCA, that we call {\em shift} (by $\bm{x}$), $\tau^{\bm x}(O_{\bm y})=O_{\bm x+\bm y}$ \footnote{We used the additive notation which is suitable only if translations on the lattice are abelian. This is the case for the present Letter, even though more general situations can be conceived~\cite{ARRIGHI_MARTIEL_NESME_2018,DAriano:2016aa,Perinotti_2020}}.

It can be shown \cite{schumacher2004reversible} that for any $\Lambda \in \left[\Z^s\right]$ there exists a unitary $U_{\Lambda\cup(\Lambda+\cN)} \in \cA_{\Lambda\cup(\Lambda+\cN)}$ so that for any $O_\Lambda \in \cA_{\Lambda}$
\begin{align}
\alpha(O_\Lambda)=U_{\Lambda\cup(\Lambda+\cN)}^\dag\left( O_\Lambda  \otimes I_{\Delta\Lambda } \right)U_{\Lambda\cup(\Lambda+\cN)}\,,
\label{eq:locrul}
\end{align}
 with $\Delta\Lambda \coloneqq \left(\Lambda + \cN\right) \setminus \Lambda $ the boundary of region $\Lambda$.

As shown in Ref.~\cite{schumacher2004reversible}, the \emph{global} update rule $\alpha$ is uniquely determined by the \emph{local} update rule which is defined as the restriction of the former to one site
$ \alpha_0: \mathcal{A}_{\bm{0}}  \to \mathcal{A}_{\bm{0}+\mathcal{N}}\,$.
 Moreover, $\alpha_0$ defines the local update rule of a QCA $\alpha$ if and only if for all $\bm{0} \neq \bm{x} \in \mathbb{Z}^s$ it holds
\begin{equation}\label{iff local rule}
    \left[\alpha(\cA_{\bm{0}}), \alpha(\cA_{\bm{x}})\right]=\left[\alpha_0(\cA_{\bm{0}}),\,\tau^{\bm{x}}(\alpha_0(\cA_{\bm{0}}))\right]=0\,,
\end{equation} 
where $\alpha(\cA_{\bm x})$ is the image of $\cA_{\bm x}$ via $\alpha$ \footnote{We say that two algebras commute if all their operators commute.}.
The above condition is non-trivial for those 
$\bm x$ whose neighborhood overlaps with that of ${\bm0}$.

The non-constructive nature of the QCA definition reflects the absence of a general method for constructing instances of this model.
Underlying quantum computation is the idea of fragmenting a complicated unitary evolution into a series of disjoint quantum gates.
Hence, characterizing QCA executable as a finite sequence of layers of unitary block transformations---a so-called Finite-Depth Quantum Circuit (FDQC)---is highly relevant. A sufficient (but not necessary) condition for this is that the unitary matrices $U_{\bm x \cup(\bm x +\cN)}$ defining the local QCA action as in~\eqref{eq:locrul} pairwise commute~\cite{GNVW}. We stress that FDQCs provide a class of QCA but not all QCA can be realized this way.  The prototypical example of a QCA that is not an FDQC is the shift on $\mathbb Z$ \cite{GNVW}. The following classification provides the FDQC for those QCA that admit one, and identifies those that do not.
\begin{definition}
    Let $\left\{\Lambda_j^l\right\} \ss \left[\Z^s\right] $ with $ l=1,\ldots,L$ be $L$ partitions of the lattice, and let
      $R\in \mathbb{N}_{+}$ such that $  \abs{\Lambda_j^l} \leq~ R \; \forall\, l,j$. A QCA $\alpha$ is a \emph{Finite-Depth Quantum Circuit (FDQC)} of depth $L$ if it can be written as 
     \begin{equation*}
         \alpha=\upsilon_L \cdot \upsilon_{L-1} \cdot  \ldots \cdot  \upsilon_1\,,
     \end{equation*} 
     where each $\upsilon_l$ is a product of disjoint unitary transformations (\emph{gates}) localized on the sets of the $l^{th}$ lattice partition 
     \begin{equation*}
         \upsilon_l(\cdot)=\left( \prod_{ j} U_{\Lambda_j^l}^{\dag}\right)\cdot \left( \prod_{j} U_{\Lambda_j^l}\right)\,.
     \end{equation*}
\end{definition}
\noindent
\red{Notice that we use the acronym FDQC to denote realizations \emph{without ancillary systems}~\footnote{It has been proved that allowing for an ancillary copy of the system \red{one can \emph{simulate} \emph{any} QCA via a finite-depth quantum circuit~\cite{ARRIGHI2011372,arrighi2012intrinsically,arrighi2019overview}. Specifically, given any QCA $\alpha$ on $\cA_{\Z^s}$, the QCA $\alpha \otimes \alpha^{-1}: \cA_{\Z^s} \otimes \cA_{\Z^s} \to \cA_{\Z^s} \otimes \cA_{\Z^s} $, which acts as $ A \otimes B \mapsto \alpha(A) \otimes \alpha^{-1}(B)$, is realizable as an FDQC. This circuit simulates the automaton $\alpha$ on the first copy of the system.}}}.

In the following we assume $\mathcal{N}$ to be the von Neumann neighborhood scheme (FIG.~\ref{fig:neighborhood}) $\mathcal{N}\subseteq\{\bm{0}\} \cup  \{ \pm \bm{e}_i \}_{i=1}^s$, where $\{ \bm{e}_i \}_{i=1}^s$ is the canonical basis of $\Z^s$. 
\begin{figure}[h]
    \centering
    \includegraphics[width=6.2cm]{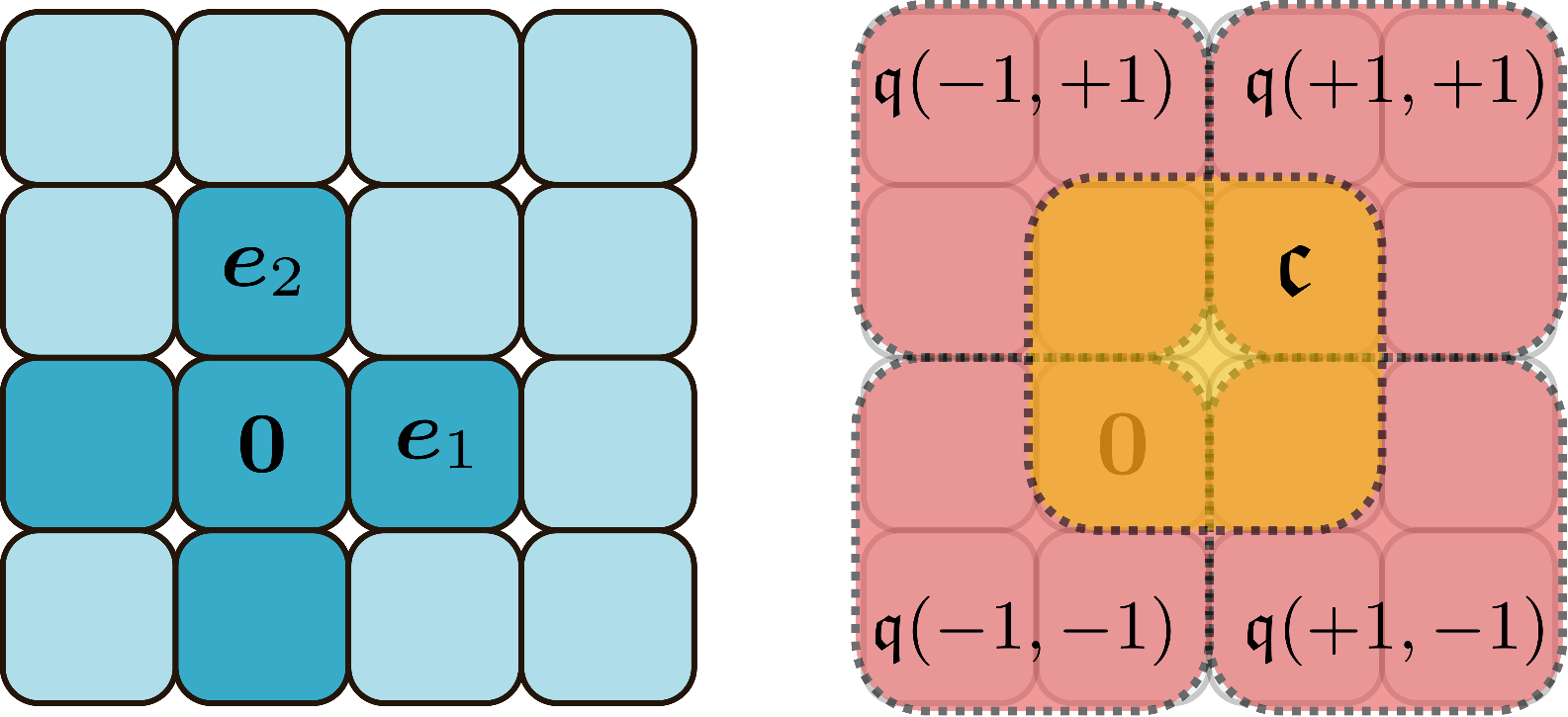}
    \caption{\red{On the left, a darker tone shows the von Neumann neighborhood scheme on $\Z^2$. Cells are labelled by the Cartesian coordinates of their center. On the right are depicted the super-cell $\mathfrak{c}$ and quadrant sets $\mathfrak q(\bm q)$ on $\Z^2$.}
    }
    \label{fig:neighborhood}
\end{figure}
Pauli matrices are denoted by $\bm{\sigma}=(\sigma^1, \sigma^2, \sigma^3)$.
The first result is the following (FIG.~\ref{fig:classqubit}).
\begin{theorem}\label{thm:classific}
The local update rule $\alpha_0 $ of a QCA $\alpha$ with von Neumann neighborhood $\cN$ of a hypercubic lattice $\mathbb{Z}^{s}$ of qubits is given by
 \begin{equation*}
     \alpha_0(O)=U^\dag \left(  V^\dag O V \otimes I_{\Delta\bm{0}}   \right) U \quad \forall\,O \in \cA_{\bm{0}}\,,
 \end{equation*}
 where \red{$V \in \mathsf{SU}(2)$ is a site-wise unitary }
 \begin{equation*}%\label{main}
 %\begin{gathered}
     V(\bm{\theta})=R_z(\theta_1)R_y(\theta_2) R_z(\theta_3)
%     \\
%     \text{with}\quad 
%     R_j(\lambda) \coloneqq \exp{-i \lambda\, \sigma^j /2} \quad \lambda\in [0,2\pi]\,,
%     \end{gathered}
 \end{equation*}
with $R_j(\lambda) \coloneqq \exp{-i \lambda\, \sigma^j /2}$, $\lambda\in [0,2\pi]$, while $U\in \mathsf{U}\left(2^{\abs{\cN}}\right)$ can be
 \begin{itemize}
     \item a shift
     \begin{equation*}%\label{shift}
      \begin{gathered}
          \tau^{\bm{y}}\lvert_{\cA_{\bm{0} }} =  I_{\Delta\bm{0}\setminus \bm{y}} \otimes S_{\bm{0}, \bm{y}}\,,    
     \end{gathered} 
     \end{equation*}
     $\text{with the swap matrix} \quad S_{\bm{0}, \bm{y}}=\sum_{ab} \ketbra{ab}{ba}_{\bm{0}\bm{y}}$
     \item a multiply controlled-phase
    \begin{equation*}%\label{controlled phase}
        \begin{aligned} 
                   &M (\bm{\varphi}) \coloneqq\prod_{i=1}^s C_{\bm 0,\bm e_i}(\varphi_i)\; C_{\bm 0,-\bm e_i}(\varphi_i)\;,\\
                   &C_{\bm x,\bm y}(\varphi)\coloneqq\, \sum_{a  = 0,1 }\; \dyad{a}_{\bm{x}}\otimes  \exp{\frac{i}{2}(\,I-\sigma^3\,)_{\bm{y}}\, a\, \varphi }\,, 
     \end{aligned} 
     \end{equation*}
     $\text{with} \;\; \varphi_i \in [0,2\pi]\,.$
 \end{itemize}
\end{theorem}
\begin{figure}[h]
    \centering
    \includegraphics[width=8.5cm]{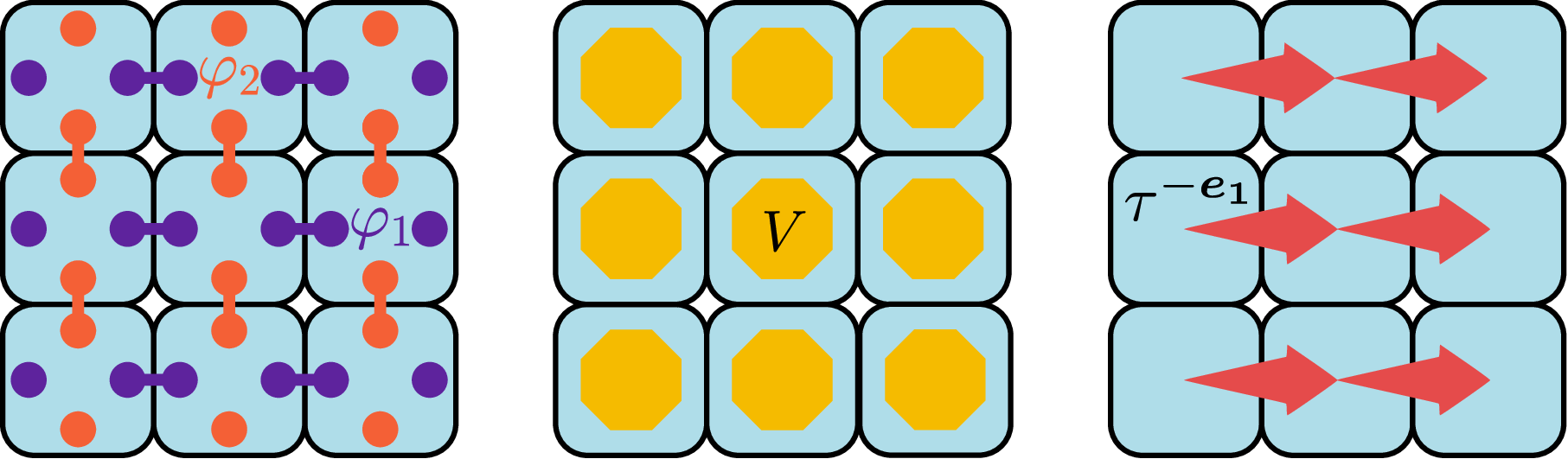}
    \caption{\red{Update rules for qubit QCA with von Neumann neighborhood}. From left to right: multiply controlled-phase, site-wise unitary, and shift QCA on $\Z^2$.}
    \label{fig:classqubit}
\end{figure}
\noindent
Since multiply controlled-phases decompose into commuting two-qubit controlled-phase gates $C_{\bm x,\bm y}(\varphi_i)$, \red{they can be implemented in $2^s$} sequential layers, so that no two-qubit gates in the same layer overlap. Site-wise unitaries can be included in the \red{first} layer. Conversely, as swap gates do not commute, implementing a shift QCA on a closed chain of $N$ qubits as a circuit requires $N-1$ layers, each with a single swap gate, resulting in a circuit depth that scales linearly with system size. 

The primary tools employed for the classification Theorem \ref{thm:classific} are the so-called \emph{support} or \emph{interaction algebras} \cite{GNVW,PhysRevA.63.012301,supp}. \red{Let $\Lambda \supseteq \cN$}, the support algebras of $\alpha(\cA_{\bm x})$ on the regions $\Omega_i$ of a partition $\{ \Omega_i \}$ of $\bm x +\Lambda$  are the smallest C*-subalgebras $\cS^{\bm x}_{\Omega_i}\subseteq\cA_{\Omega_i}\se\cA_{\bm x +\Lambda}$, such that
\begin{equation}\label{def:suppalg}
    \alpha(\cA_{\bm x}) \se \bigotimes_i \cS^{\bm x}_{\Omega_i} \se \cA_{\bm x + \Lambda}\,.
\end{equation}
Support algebras allow for tracking how the image of an algebra distributes within its neighborhood.
If the images of two algebras $\alpha(\cA_{\Lambda_1}),\, \alpha(\cA_{\Lambda_2})$ commute over a region $\Omega$, the respective support algebras on that region will likewise commute $[ \cS^{\Lambda_1}_{\Omega},\, \cS^{\Lambda_2}_{\Omega}]=0$~\cite{GNVW}.  This is the crucial property that, along with translation invariance, enables for the classification of the admissible configurations of support algebras on the neighborhood of a cell.
For instance, for $\Z^2$ we proved that either
 \begin{align}
           &\alpha(\cA_{\bm 0}) \se\; \begin{matrix}
                    & & \cI & &\\
                    & & \otimes & &\\
                    \cI & \otimes & \cI & \otimes & \cM_2\\
                    & & \otimes & &\\
                    & & \cI & &
                \end{matrix}\,,
                \label{eq:supp_shift}
\end{align}
up to permutations of the neighboring sites, or
\begin{align}
                    &\alpha(\cA_{\bm 0}) \se \begin{matrix}
                    & & \cD(\bm{n}_2) & &\\
                    & & \otimes & &\\
                     & \cD(\bm{n}_1) & \otimes\;\;\cM_2 \;\;\otimes&   \cD(\bm{n}_1)\\
                    & & \otimes & &\\
                    & & \cD(\bm{n}_2) & &
                \end{matrix}, \label{eq:supp_phase}
        \end{align}
        where the support algebras on the r.h.s. are represented as if they were located at lattice sites, and $\cD(\bm{n})$ denotes the abelian algebra of jointly diagonalizable matrices
\begin{equation*}
    \cD(\bm{n}) \coloneqq \left\{ a\,I+b\,\bm{n} \cdot \bm{\sigma}  \; \mid \;  a,b \in \C \;;\; \bm{n} \in \mathbb{R}^3\,,\; \norm{\bm{n}}_2=1 \right\}\,.
\end{equation*} 
This result \red{generalizes to $\mathbb{Z}^s$ \cite{supp}.}
 % \begin{enumerate}
 %         \item[\textit{i.}] either $ \alpha(\cA_{\bm 0}) \se \cA_{\bm{x}} \otimes \cI_{\cN \setminus \bm{x}}$, with $\bm{x}\in\cN\,$;
 %         \item[\textit{ii.}] or $\alpha(\cA_{\bm 0}) \se \cA_{\bm{0}} \otimes \bigotimes_{i=1}^s(\cD(\bm{n}_i)_{-{\bm e}_i}{\otimes}\,\cD(\bm{n}_i)_{{\bm e}_i})\,$.
 %     \end{enumerate}
  Leveraging this preliminary classification we determine all the unitary matrices realizing a legitimate local update rule.
Here we provide an outline of the proof of Theorem~\ref{thm:classific} as follows \cite{supp}. 
 \red{In case \eqref{eq:supp_shift}, the local update rule $\alpha_0$ realizes a site-wise unitary, possibly combined with a shift of the algebra $\cA_{\bm{0}}$ by one site along $\pm \bm{e}_i$.}
  % In case \eqref{eq:supp_shift}, the local rule $\alpha_0$ must realize either a shift of the algebra $\cA_{\bm{0}}$ by one site along $\pm \bm{e}_i$, possibly combined with a site-wise unitary, or just a site-wise unitary.
 In case \eqref{eq:supp_phase}, modulo a basis change, we can impose the necessary and sufficient commutation relations (\ref{iff local rule}) on the image of any operator $O \in \cA_{\bm{0}}$ and $O'\in\cA_{\bm x}$ for any $\bm x\in\{\pm \bm e_i,\pm(\bm e_i+\bm e_j),\pm(\bm e_i-\bm e_j)\}$. The image of an operator 
 $O \in \cA_{\bm{0}}$ is of the form
\begin{equation*}
     \alpha_0(O)=\sum_{\bm a,\bm b}\dyad{\bm a}_+\otimes                    
                    U^\dag_{\bm a\bm b}\, O\, U_{\bm a\bm b}\otimes\dyad {\bm b}_{-},
\end{equation*}
where $\dyad{\bm c}_{\pm}=\bigotimes_{j=1}^s\dyad{c_j}_{\pm{\bm e}_j}$.
For fixed values of $a_i,b_i$ in all coordinate directions except one, the local rule must obey the same constraints as that of a QCA on $\mathbb Z$, which, according to Ref.~\cite{schumacher2004reversible}, is given by a one-dimensional multiply controlled-phase in the directions $\bm{e}_i,-\bm{e}_i$. The careful matching of the conditions in various directions leads to the thesis. 

Being algebra automorphisms,  
QCA are fully specified by their action on a set of generators of the qubit algebra \footnote{An algebra consists of all finite linear combinations over $\C$ of finite products of its generators.}, e.g. $\{\sigma^1, \sigma^2\}$. For the subclass of Clifford QCA, 
up to site-wise unitaries by multiples of $\pi/2$  and shifts, the local rule is given by multiply controlled-phases with angles $\varphi_i \in \{0,\pi\}$, that is, for $j=1,2\,$
\[ \alpha_0(\sigma^j_{\bm{0}})= \sigma^j_{\bm{0}}\; \otimes\bigotimes_{\bm{x}\in\{ \pm \bm{e}_i \}} (\sigma^3_{\bm{x}})^{\frac{\varphi_i}\pi}.
\]

We now provide a further classification based on the information flow through a boundary. 
Here as well, support algebras are the cornerstone of our analysis.
Let us introduce some notation (FIG.~\ref{fig:neighborhood}) and recall the result for $1+1$ dimension. Let $\mathfrak{c}$ be the \textit{super-cell} 
\begin{align*}
  \mathfrak{c} \coloneqq \{  \bm{x} \in \Z^s \; \mid \;  x_i \in \{ 0,1\} \}
\end{align*}
consisting of $2^s$ sites, $Q$ the set of $2^s$ \textit{quadrants} $\mathfrak q(\bm q)$
\begin{align*}
Q \coloneqq \left\{ \mathfrak q(\bm q)=\mathfrak{c} + \{\bm{q}\}  \; \mid \; q_i \in \{\pm 1\}  \right\}\,,  
\end{align*}
and $d(\Lambda)$ the square root dimension of the matrix algebra $\cM_{d(\Lambda)}$ supported on $\Lambda \in [\Z^s]$. 
Any nearest-neighbor QCA $\alpha$ on a lattice $\Z$ of qudits $\cA_x \cong \cM_{d}$, 
restricts to an isomorphism
\begin{equation}\label{isom 1d}
 \alpha(\cA_{\mathfrak{c}} ) = \bigotimes_{\mathfrak q(q) \in Q} \cS^{\mathfrak{c}}_{\mathfrak q(q)}  \cong \bigotimes_{\mathfrak q(q) \in Q} \cM_{d(\mathfrak q(q))}\,,   
\end{equation}
with $d\left(\mathfrak q(-1)\right)\,d\left(\mathfrak q(+1)\right)=d(\mathfrak{c})$, and its \emph{index}
\[ \iota(\alpha) \coloneqq  \frac{d(\mathfrak q(-1))}{d}\]
uniquely determines a topological invariant that identifies equivalence classes of QCA modulo concatenation with FDQCs~\cite{GNVW}~\footnote{Here the result is stated for translation-invariant QCA, but it can be proved more generally under the hypothesis of locality~\cite{GNVW}, \red{and for multiple $1$ dimensional chains linked together~\cite{freedman2020classification}.}}. A \red{one-dimensional} QCA is an FDQC if and only if its index is $1$, \red{requiring only two layers}; otherwise, it is equivalent to a shift.
%—a realization scheme known as the \emph{Margolus FDQC} (FIG.~\ref{fig:Margolus}), defined below for any spatial dimension.
The index quantifies the net flow of information through the local systems at a timestep: the index for a right (left) shift is $\iota(\tau^{-e_1})=d$ ($\iota(\tau^{e_1})=d^{-1}\,$). A QCA with null flux has index $1$ like the trivial evolution.

 We can define analogs of the one-dimensional index along each Cartesian direction of $\Z^s$ (FIG.~\ref{fig:index}).
The index of the automaton $\alpha$ is an $s$-dimensional vector $\bm{\iota}(\alpha)$ with components
% \begin{equation}\label{index comp}
%         \iota_{i}(\alpha) \coloneqq \left(\frac{ \prod_{\mathfrak q(\bm q) \in R_{i}} d\left(\mathfrak q(\bm q)\right)} { d^{2^{s-1}} } \right)^{1/2^{s-1}} 
% \end{equation}
\begin{equation}\label{index comp}
        \iota_{i}(\alpha) \coloneqq \sqrt[\leftroot{-3}\uproot{3}2^{s-1}]{ \frac{ \prod_{\mathfrak q(\bm q) \in R_{i}} d\left(\mathfrak q(\bm q)\right)} { d^{2^{s-1}} } } 
\end{equation}
where $R_{i} \coloneqq \left\{ \mathfrak q(\bm q) \in Q\mid q_i=-(\bm{e_i})_i \right\} $, and $d=2$ for qubits. For our qubit QCA, identity (\ref{isom 1d}) holds in any spatial dimension, and the classification generalizes that of the \red{one-dimensional case} (Theorem~\ref{th:four}).
% QCA realizable as an FDQC are precisely those with all indices equal to $\red{0}$, \red{and they can be realized with two layers (FIG.~\ref{fig:neighborhood}) (recasting controlled-phases and site-wise unitaries into $2^s$-qubit gates). } The remainders are equivalent to shifts.  
% Now, e.g., a right shift on $\Z^2$ has index $\bm \iota (\tau^{-\bm{e_1}})=(2,1)$. 

\begin{figure}[h]
    \centering
    \includegraphics[width=7.0cm]{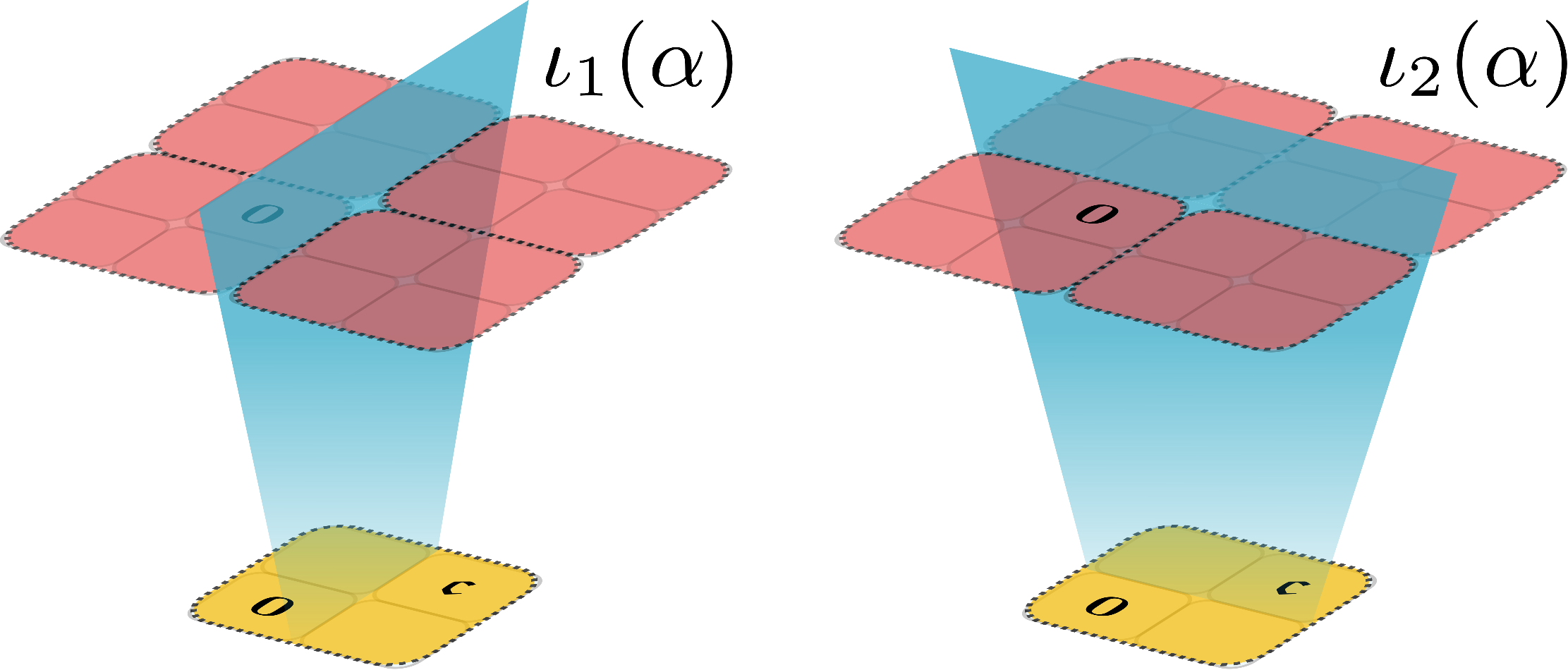}
    \caption{\red{On a $\Z^2$ lattice}, indices $\iota_1(\alpha)$ and $\iota_2(\alpha)$ measure the portion of information stored in the yellow algebra $\cA_{\mathfrak{c}}$ that the update rule $\alpha$ rearranges and transmits to the right or left and up or down sectors of the red algebra $\alpha(\cA_{\mathfrak{c}})$, respectively.}
    \label{fig:index}
\end{figure}

\red{On the contrary, for a generic QCA of $\cA_{\Z^s}$, the indices defined in (\ref{index comp}) do not provide a thorough classification modulo FDQCs as \red{an inclusion (Eq.~\eqref{def:suppalg}) replaces} 
% \begin{equation*}
%      \beta(\cA_{\mathfrak{c}} ) \se \bigotimes_{\mathfrak q(\bm q) \in Q} \cS^{\mathfrak{c}}_{\mathfrak q(\bm q)}\,,
% \end{equation*}
the identity of Eq.~\eqref{isom 1d}.
%with $\cS^{\mathfrak{c}}_{\mathfrak q(\bm q)}$ isomorphic to a direct sum of matrix algebras.
}
The latter is derived by exploiting the simplicity of \red{$\cA_{\Z}$}~\footnote{An algebra $\cA$ is called simple if its closed two-sided ideals are only $\{0\}$ and $\cA$. A matrix algebra with a trivial center is simple. A quasi-local algebra of simple algebras is simple \cite{brattelirobinson}.}, and the commutativity of support algebras on quadrants, which generally 
% , which holds for our qubit automata \red{in arbitrary spatial dimensions} and allows us to extend the index theory to them, 
fails \red{in higher spatial dimensions}. If we assume it from the outset, i.e.~imposing
    \begin{equation}\label{Vogts hyp}
       \left[ \tau^{-\bm{q}}\,\cS_{\mathfrak q(\bm q)}^{\mathfrak{c}},\, \tau^{-\bm{p}}\, \cS_{\mathfrak q({\bm{p}})}^{\mathfrak{c}} \right]=0\quad \forall\,\mathfrak q(\bm q),\mathfrak q({\bm{p}}) \in Q\,, 
    \end{equation}
we can prove \red{that Eq.~\eqref{isom 1d} holds}
% \begin{equation*}
%         \alpha(\cA_{\mathfrak{c}} ) =\bigotimes_{\mathfrak q(\bm q) \in Q} \cS^{\mathfrak{c}}_{\mathfrak q(\bm q)} \cong \bigotimes_{\mathfrak q(\bm q)\in Q} \cM_{d(\mathfrak q(\bm q))}\,,
%     \end{equation*}
for any QCA of a lattice $\Z^s$ of qudits $\cA_x \cong \cM_{d}$ with a neighborhood enclosed within the Moore scheme $\cN \se \{\bm{0}\} \cup  \{ \pm \bm{e_i} \}_{i=1}^s \cup \{ \bm{e_i} \pm \bm{e_j} \}_{i,j=1}^s \cup \{ -\bm{e_i} \pm \bm{e_j} \}_{i,j=1}^s$,      with $\prod_{\mathfrak q(\bm q)\in Q}d(\mathfrak q(\bm q)) =d(\mathfrak{c})= d^{2^s}  $  \cite{supp}.
\red{Nevertheless, condition (\ref{Vogts hyp}) is only necessary--but not sufficient--to guarantee that the indices will provide a meaningful classification in dimension $s\geq2$}.
Our qubit QCA satisfy condition (\ref{Vogts hyp}) and their indices are shown to classify them strictly. Then, we may introduce a class of QCA that is their natural extension and that can be classified according to indices in the same way, \red{by adding the extra assumption that the dimension of the cell is prime. This leads to} the following Theorem \ref{th:four}, whose proof can be found in the Supplemental Material~\cite{supp}. 
\begin{theorem}\label{th:four}
Let $\mathscr{V}_p^s$ be the sets of von Neumann neighborhood QCA of hypercubic lattices $\Z^s$ with $\cA_{\bm{x}} \cong \cM_p$, $p$ prime, satisfying Equation (\ref{Vogts hyp}). The index $\bm{\iota}$ classifies QCA in $\mathscr{V}_p^s$ modulo FDQCs: 
$\alpha \in \mathscr{V}_p^s$ is an FDQC if and only if $\bm{\iota}(\alpha)=\bm{1}$, while all other QCA are equivalent to shifts.
\red{Any such FDQC can be realized with two layers, acting on lattice partitions defined by even translates of the super-cell and quadrants (FIG.~\ref{fig:neighborhood})}.
\end{theorem}

While some QCA are universal for quanutm computing, i.e.~equivalent to a quantum Turing machine or the circuit model \cite{arrighi2019overview}, this might not be the case for our qubit QCA. However, they generate cluster states efficient as a universal resource (preparator) for Measurement-based Quantum Computation ~\cite{van2007fundamentals,gross2007measurement, PhysRevA.110.062617}. It is then interesting to study how those QCA process entanglement.

We simulated the most general QCA $\alpha^{\bm{\varphi},\bm{\theta}} \in \mathscr V_2^s$  with an FDQC realization, i.e.~site-wise unitary followed by  controlled-phase gates (FIG.~\ref{fig:circuit}). 
\begin{figure}[h]
    \centering
    \includegraphics[width=7.4cm]{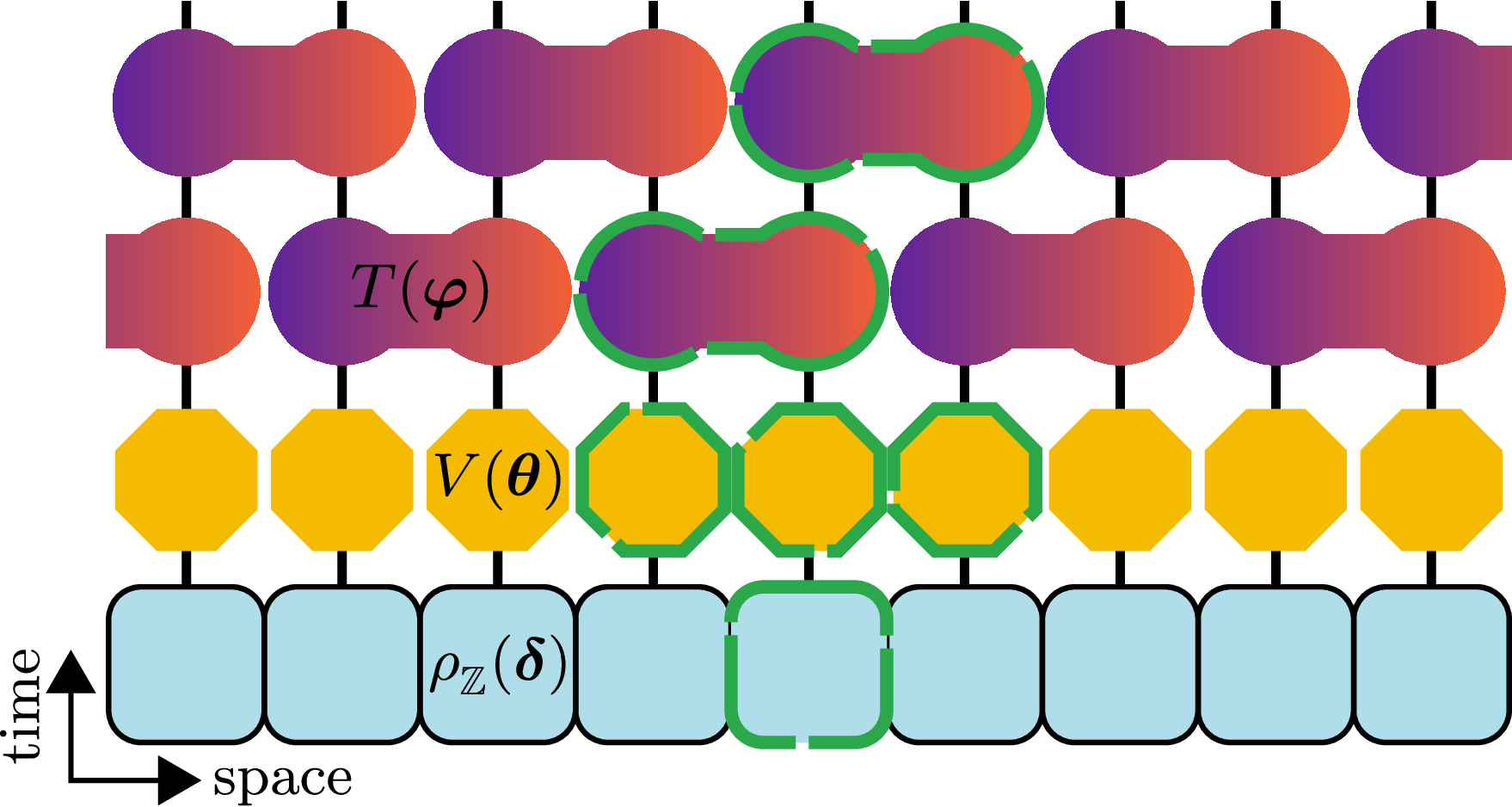}
    \caption{One time-step of the QCA $\alpha^{\bm{\varphi},\bm{\theta}}$ on $\Z^2$. States are depicted as squares, site-wise unitaries as octagons, and tiled controlled-phases $T(\bm \varphi) \coloneqq C_{\bm x, \bm x + \bm e_1}(\varphi_1)\, C_{\bm x + \bm e_2, \bm x + \bm k }(\varphi_1)\, C_{\bm x, \bm x + \bm e_2}(\varphi_2)\, C_{\bm x + \bm e_1, \bm x + \bm k }(\varphi_2)$, with $\bm k \coloneqq \bm e_1 + \bm e_2\,$, are shaded to remind that phases can be different in different lattice directions. Angles $\bm{\varphi} \in [0,2\pi]^s$ and $ \bm{\theta},\bm{\delta} \in [0,2\pi]^3$ are free parameters of the circuit. Transformations within the future causal cone of the system spotted by a dashed frame are highlighted in the same way.}
    \label{fig:circuit}
\end{figure}
To probe entanglement generation, we computed the entanglement entropy of a single qubit versus the parameters of the automaton and the input state. The latter is a pure separable translation-invariant state $\rho(\bm{\delta}) = \bigotimes_{\bm{x}\in\left[ \Z^s\right]} \rho_{\bm{x}}(\bm{\delta})$ with $ \rho_{\bm{x}}(\bm{\delta})=V(\bm{\delta}) \dyad{0}_{\bm{x}} V^\dag(\bm{\delta})$. The entanglement entropy of the qubit at the origin at step $t \in \mathbb{N}$ is defined by 
\begin{equation*}
S\left(\,\rho_{\bm{0}}^t(\bm{\varphi},\bm{\theta},\bm{\delta})\,\right) \coloneqq -\mathrm{Tr}\left(\,\rho_{\bm{0}}^t \,\log_2\,\rho_{\bm{0}}^t\,\right)\,,   
\end{equation*}
where the reduced state $\rho_{\bm x}\in\cM_2$ is obtained imposing
\begin{align*}
\Tr[\rho_{\bm x}\, O]=\Tr[\rho\, O_{\bm x}]\,, \quad \forall O_{\bm x}\in\mathcal A_{\bm x}\,,
\end{align*}
$O$ being the operator representing $O_{\bm x}$ in the abstract algebra $\cM_2$, and the superscript $t$ reminds that the state is considered at step $t$.
As the number of sites in the causal future of a qubit is exponential versus $t$, classical simulation becomes prohibitive as early as $t \geq 4$ for $s=2$ or $t \geq 3$ for $s=3$. We simulated QCA up to $3$ and $2$ time-steps on $\Z^2$ and $\Z^3$ lattices, respectively \footnote{The algorithm was developed using the Qiskit library (version 0.45.3), and its performances match those reported in \cite{suzuki2021qulacs}. When employing JAX and Pennylane libraries (versions 0.4.24 and 0.34.0), we encounter a significant overhead as the number of qubits and parameters grows, resulting in a less efficient implementation.}. We computed the following functions, quantifying entanglement in terms of the QCA parameters: 
\begin{align}
    &S^{\varphi_1,\theta_2}_{\mathrm{max}} \coloneqq \max_{\theta_1,\bm{\delta}} S\left(\,\rho_{\bm{0}}^t(\bm{\varphi},\theta_1,\theta_2,\bm{\delta})\,\right)\,, \nonumber\\ 
    &\Delta S^{\varphi_1,\theta_2} \coloneqq S^{\varphi_1,\theta_2}_{\mathrm{max}} - \min_{\bm{\delta}} S\left(\,\rho_{\bm{0}}^t(\bm{\varphi},\tilde\theta_1,\theta_2,\bm{\delta})\,\right)\,, \label{ds}
\end{align}
with $\tilde\theta_1 \in \mathrm{argmax}_{\theta_1}\left[\max_{\bm{\delta}} S(\,\rho_{\bm{0}}^t(\bm{\varphi},\bm{\theta},\bm{\delta})\,) \lvert_{\theta_1}\right]$, $\varphi_i=\varphi_j$ for all $i,j$, 
%\begin{align*}
%& \tilde\theta_1 \in \mathrm{argmax}_{\theta_1}%[\max_{\bm{\delta}} %S\left(\rho_{\bm{0}}^t(\bm{\varphi},\bm{\theta},\bm{\delta})\righ%t) \lvert_{\theta_1}], 
%\end{align*}
and 
\begin{equation*}
    S^{\varphi_1,\varphi_2}_{\mathrm{max}} \coloneqq \max_{\bm{\theta},\bm{\delta}} S\left(\,\rho_{\bm{0}}^t(\bm{\varphi},\bm{\theta},\bm{\delta})\,\right)\,,
\end{equation*}
%\begin{equation*}\label{s2}
%    S^{\varphi_1,\varphi_2}_{\mathrm{max}} \coloneqq %\max_{\bm{\theta},\bm{\delta}} %S\left(\rho_{\bm{0}}^t(\bm{\varphi},\bm{\theta},\bm{\delta})\righ%t) \,,
%\end{equation*}
with $\varphi_3=\varphi_2$ (FIG.~\ref{fig:merged}). 
By manipulating the expressions, one can always discard the parameters $\theta_3, \delta_3$, and for two time steps also $\theta_1$.
The behavior of function (\ref{ds}) remains stable vs $\tilde\theta_1$, with fluctuations that are irrelevant for our purposes. The sample size for input states is $\sim 50$. 

For the vast majority of parameter choices, the QCA can generate (approximately) maximum entanglement. The blue regions in Figure \ref{fig:merged}(b) that are red in Figure \ref{fig:merged}(a), identify QCA where this happens independently of the input state. The extent of the latter regions increases versus number of time steps and lattice dimension. Conversely, the output states are separable for $\varphi_i=2k\pi $, as well as for $(\varphi_i,\,\theta_2)=(2k\pi/t,\,k\pi)$, with $k \in \Z$. The first case corresponds to switching off the controlled-phases (the only entangling gates). In the second case, for $\theta_2=2k\pi$ we are removing site-wise unitaries, with the angles of controlled phases adding up to $2\pi$. When $\theta_2 = (2k+1)\pi $, site-wise unitaries become NOT gates, which just flip the computational basis. 
In the neighborhood of these sets that leave the state factorized, the automata generate entanglement at a polynomial rate.
% $(0,k\pi)$.

% \begin{figure}[H]
%     \centering
%     \includegraphics[width=8.6cm]{phivsphi.eps}
%     \caption{Maximum entanglement entropy of a qubit as a function of controlled-phase angles $\varphi_1$ and $\varphi_2(=\varphi_3)$. The values are optimized over input states (and for the case $t=3$, also over the \red{site-wise unitary} angle $\theta_1$, which is irrelevant for the others). Entanglement grows as time and lattice size increase.}
%     \label{fig:phivsphi}
% \end{figure}
% \begin{figure}[h]
%     \centering
%     \includegraphics[width=8.6cm]{full_phivstheta2.eps}
%     \caption{On the left (panels (a)), the maximum entanglement entropy of a qubit is plotted against the controlled-phase angle $\varphi_1=\varphi_i\; \forall\,i$ and the \red{site-wise unitary} angle $\theta_2$. The values are optimized over input states. For the case $t=3$, the optimization also runs over the angle $\theta_1$, which is irrelevant for the others. On the right side (panels (b)), we display the difference between the maximum and minimum entanglement entropy of a qubit for the same automata. 
% The generated entanglement increases 
% with time 
% and with spatial dimension.}
%     \label{fig:full}
% \end{figure}

\begin{figure}[h]
    \centering
    \includegraphics[width=8.6cm]{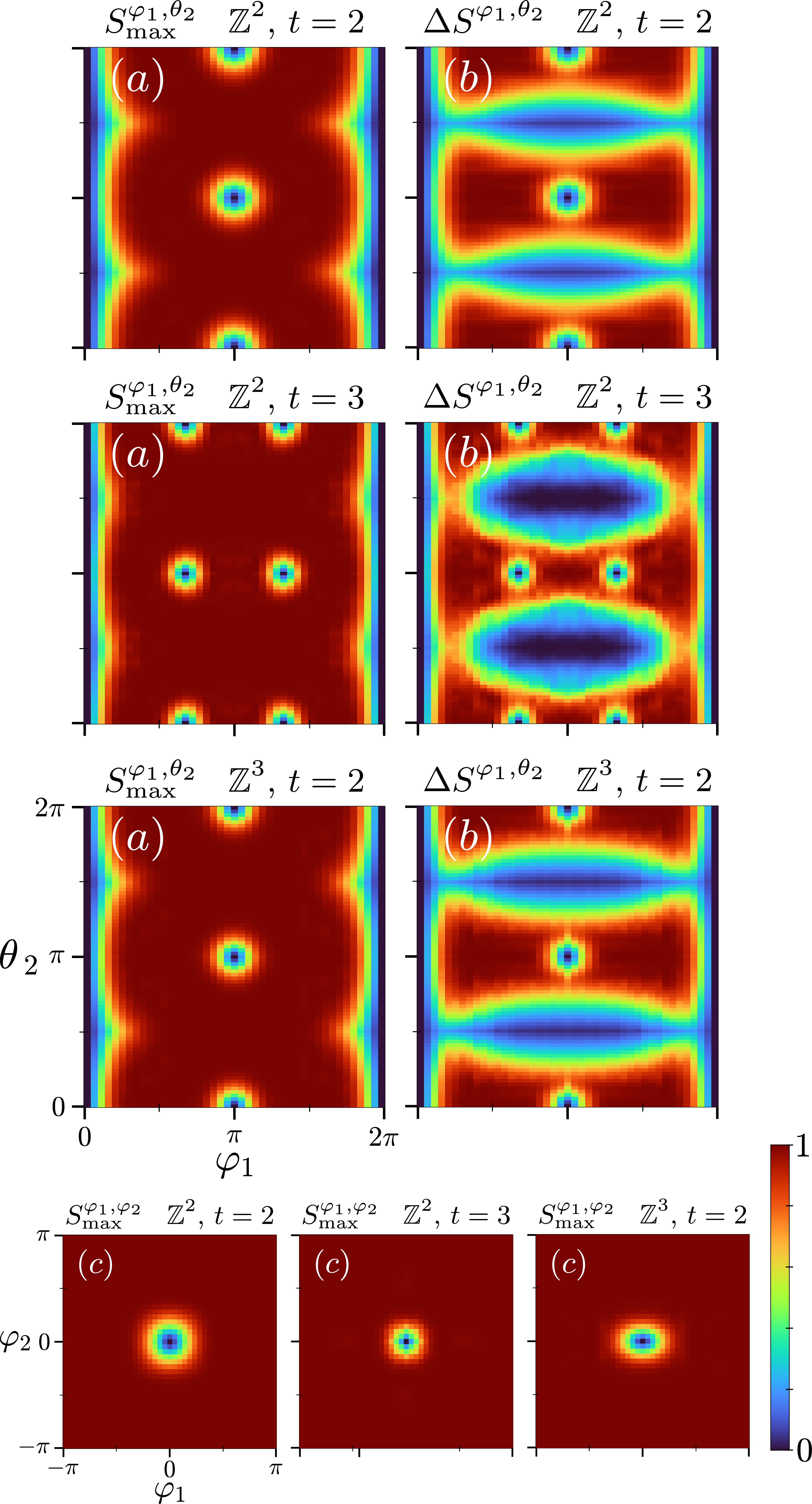}
    \caption{
    Panels (a): maximum entanglement entropy of a qubit against controlled-phase angle $\varphi_1=\varphi_i\; \forall\,i$ and \red{site-wise unitary} angle $\theta_2$. Panels (b) display the difference between the maximum and minimum entanglement entropy of a qubit for the same automata. Panels (c): maximum entanglement entropy of a qubit as a function of controlled-phase angles $\varphi_1$ and $\varphi_2(=\varphi_3)$. Values are optimized over input states. For the case $t=3$, the optimization also runs over the angle $\theta_1$, which is irrelevant for $t=2$. Generated entanglement increases with time and with spatial dimension.}
    \label{fig:merged}
\end{figure}

%----------------------------------------------------------------------------------------------------------
%----------------------------------------------------------------------------------------------------------

\acknowledgments \textit{Acknowledgments}—PP and AB acknowledge financial support from the European Union—Next Generation E.U. through the MUR project Progetti di Ricerca
d’Interesse Nazionale (PRIN) QCAPP No. 2022LCEA9Y. AB acknowledges financial support
from the European Union—Next Generation E.U. through
the PNRR MUR Project No. PE0000023-NQSTI. AP
acknowledges scientific discussion with E.~Centofanti, D.~Cugini, F.~Scala, and L.~S.~Trezzini.

\bibliographystyle{apsrev4-1}
\bibliography{qubit_classification}

\providecommand{\noopsort}[1]{}\providecommand{\singleletter}[1]{#1}%
\begin{thebibliography}{51}%
\makeatletter
\providecommand \@ifxundefined [1]{%
 \@ifx{#1\undefined}
}%
\providecommand \@ifnum [1]{%
 \ifnum #1\expandafter \@firstoftwo
 \else \expandafter \@secondoftwo
 \fi
}%
\providecommand \@ifx [1]{%
 \ifx #1\expandafter \@firstoftwo
 \else \expandafter \@secondoftwo
 \fi
}%
\providecommand \natexlab [1]{#1}%
\providecommand \enquote  [1]{``#1''}%
\providecommand \bibnamefont  [1]{#1}%
\providecommand \bibfnamefont [1]{#1}%
\providecommand \citenamefont [1]{#1}%
\providecommand \href@noop [0]{\@secondoftwo}%
\providecommand \href [0]{\begingroup \@sanitize@url \@href}%
\providecommand \@href[1]{\@@startlink{#1}\@@href}%
\providecommand \@@href[1]{\endgroup#1\@@endlink}%
\providecommand \@sanitize@url [0]{\catcode `\\12\catcode `\$12\catcode `\&12\catcode `\#12\catcode `\^12\catcode `\_12\catcode `\%12\relax}%
\providecommand \@@startlink[1]{}%
\providecommand \@@endlink[0]{}%
\providecommand \url  [0]{\begingroup\@sanitize@url \@url }%
\providecommand \@url [1]{\endgroup\@href {#1}{\urlprefix }}%
\providecommand \urlprefix  [0]{URL }%
\providecommand \Eprint [0]{\href }%
\providecommand \doibase [0]{http://dx.doi.org/}%
\providecommand \selectlanguage [0]{\@gobble}%
\providecommand \bibinfo  [0]{\@secondoftwo}%
\providecommand \bibfield  [0]{\@secondoftwo}%
\providecommand \translation [1]{[#1]}%
\providecommand \BibitemOpen [0]{}%
\providecommand \bibitemStop [0]{}%
\providecommand \bibitemNoStop [0]{.\EOS\space}%
\providecommand \EOS [0]{\spacefactor3000\relax}%
\providecommand \BibitemShut  [1]{\csname bibitem#1\endcsname}%
\let\auto@bib@innerbib\@empty
%</preamble>
\bibitem [{\citenamefont {Schumacher}\ and\ \citenamefont {Werner}(2004)}]{schumacher2004reversible}%
  \BibitemOpen
  \bibfield  {author} {\bibinfo {author} {\bibfnamefont {B.}~\bibnamefont {Schumacher}}\ and\ \bibinfo {author} {\bibfnamefont {R.~F.}\ \bibnamefont {Werner}},\ }\href {https://arxiv.org/abs/quant-ph/0405174} {\  (\bibinfo {year} {2004})},\ \Eprint {http://arxiv.org/abs/quant-ph/0405174} {arXiv:quant-ph/0405174 [quant-ph]} \BibitemShut {NoStop}%
\bibitem [{\citenamefont {Farrelly}(2020)}]{Farrelly2020reviewofquantum}%
  \BibitemOpen
  \bibfield  {author} {\bibinfo {author} {\bibfnamefont {T.}~\bibnamefont {Farrelly}},\ }\href {\doibase 10.22331/q-2020-11-30-368} {\bibfield  {journal} {\bibinfo  {journal} {{Quantum}}\ }\textbf {\bibinfo {volume} {4}},\ \bibinfo {pages} {368} (\bibinfo {year} {2020})}\BibitemShut {NoStop}%
\bibitem [{\citenamefont {Arrighi}(2019)}]{arrighi2019overview}%
  \BibitemOpen
  \bibfield  {author} {\bibinfo {author} {\bibfnamefont {P.}~\bibnamefont {Arrighi}},\ }\href {\doibase https://doi.org/10.1007/s11047-019-09762-6} {\bibfield  {journal} {\bibinfo  {journal} {Natural Computing}\ }\textbf {\bibinfo {volume} {18}},\ \bibinfo {pages} {885} (\bibinfo {year} {2019})}\BibitemShut {NoStop}%
\bibitem [{\citenamefont {Trezzini}\ \emph {et~al.}(2025)\citenamefont {Trezzini}, \citenamefont {Bisio},\ and\ \citenamefont {Perinotti}}]{Trezzini2025renormalisationof}%
  \BibitemOpen
  \bibfield  {author} {\bibinfo {author} {\bibfnamefont {L.~S.}\ \bibnamefont {Trezzini}}, \bibinfo {author} {\bibfnamefont {A.}~\bibnamefont {Bisio}}, \ and\ \bibinfo {author} {\bibfnamefont {P.}~\bibnamefont {Perinotti}},\ }\href {\doibase 10.22331/q-2025-05-28-1756} {\bibfield  {journal} {\bibinfo  {journal} {{Quantum}}\ }\textbf {\bibinfo {volume} {9}},\ \bibinfo {pages} {1756} (\bibinfo {year} {2025})}\BibitemShut {NoStop}%
\bibitem [{\citenamefont {Rotundo}\ \emph {et~al.}(2024)\citenamefont {Rotundo}, \citenamefont {Perinotti},\ and\ \citenamefont {Bisio}}]{rotundo2024effectivedynamicsminimisingdissipation}%
  \BibitemOpen
  \bibfield  {author} {\bibinfo {author} {\bibfnamefont {A.~F.}\ \bibnamefont {Rotundo}}, \bibinfo {author} {\bibfnamefont {P.}~\bibnamefont {Perinotti}}, \ and\ \bibinfo {author} {\bibfnamefont {A.}~\bibnamefont {Bisio}},\ }\href {https://arxiv.org/abs/2412.10216} {\  (\bibinfo {year} {2024})},\ \Eprint {http://arxiv.org/abs/2412.10216} {arXiv:2412.10216 [quant-ph]} \BibitemShut {NoStop}%
\bibitem [{\citenamefont {Bisio}\ \emph {et~al.}(2021)\citenamefont {Bisio}, \citenamefont {Mosco},\ and\ \citenamefont {Perinotti}}]{PhysRevLett.126.250503}%
  \BibitemOpen
  \bibfield  {author} {\bibinfo {author} {\bibfnamefont {A.}~\bibnamefont {Bisio}}, \bibinfo {author} {\bibfnamefont {N.}~\bibnamefont {Mosco}}, \ and\ \bibinfo {author} {\bibfnamefont {P.}~\bibnamefont {Perinotti}},\ }\href {\doibase 10.1103/PhysRevLett.126.250503} {\bibfield  {journal} {\bibinfo  {journal} {Phys. Rev. Lett.}\ }\textbf {\bibinfo {volume} {126}},\ \bibinfo {pages} {250503} (\bibinfo {year} {2021})}\BibitemShut {NoStop}%
\bibitem [{\citenamefont {Bisio}\ \emph {et~al.}(2015)\citenamefont {Bisio}, \citenamefont {D'Ariano},\ and\ \citenamefont {Tosini}}]{BISIO2015244}%
  \BibitemOpen
  \bibfield  {author} {\bibinfo {author} {\bibfnamefont {A.}~\bibnamefont {Bisio}}, \bibinfo {author} {\bibfnamefont {G.~M.}\ \bibnamefont {D'Ariano}}, \ and\ \bibinfo {author} {\bibfnamefont {A.}~\bibnamefont {Tosini}},\ }\href {\doibase https://doi.org/10.1016/j.aop.2014.12.016} {\bibfield  {journal} {\bibinfo  {journal} {Annals of Physics}\ }\textbf {\bibinfo {volume} {354}},\ \bibinfo {pages} {244} (\bibinfo {year} {2015})}\BibitemShut {NoStop}%
\bibitem [{\citenamefont {Bisio}\ \emph {et~al.}(2013)\citenamefont {Bisio}, \citenamefont {D'Ariano},\ and\ \citenamefont {Tosini}}]{PhysRevA.88.032301}%
  \BibitemOpen
  \bibfield  {author} {\bibinfo {author} {\bibfnamefont {A.}~\bibnamefont {Bisio}}, \bibinfo {author} {\bibfnamefont {G.~M.}\ \bibnamefont {D'Ariano}}, \ and\ \bibinfo {author} {\bibfnamefont {A.}~\bibnamefont {Tosini}},\ }\href {\doibase 10.1103/PhysRevA.88.032301} {\bibfield  {journal} {\bibinfo  {journal} {Phys. Rev. A}\ }\textbf {\bibinfo {volume} {88}},\ \bibinfo {pages} {032301} (\bibinfo {year} {2013})}\BibitemShut {NoStop}%
\bibitem [{\citenamefont {Suprano}\ \emph {et~al.}(2024)\citenamefont {Suprano}, \citenamefont {Zia}, \citenamefont {Polino}, \citenamefont {Poderini}, \citenamefont {Carvacho}, \citenamefont {Sciarrino}, \citenamefont {Lugli}, \citenamefont {Bisio},\ and\ \citenamefont {Perinotti}}]{PhysRevResearch.6.033136}%
  \BibitemOpen
  \bibfield  {author} {\bibinfo {author} {\bibfnamefont {A.}~\bibnamefont {Suprano}}, \bibinfo {author} {\bibfnamefont {D.}~\bibnamefont {Zia}}, \bibinfo {author} {\bibfnamefont {E.}~\bibnamefont {Polino}}, \bibinfo {author} {\bibfnamefont {D.}~\bibnamefont {Poderini}}, \bibinfo {author} {\bibfnamefont {G.}~\bibnamefont {Carvacho}}, \bibinfo {author} {\bibfnamefont {F.}~\bibnamefont {Sciarrino}}, \bibinfo {author} {\bibfnamefont {M.}~\bibnamefont {Lugli}}, \bibinfo {author} {\bibfnamefont {A.}~\bibnamefont {Bisio}}, \ and\ \bibinfo {author} {\bibfnamefont {P.}~\bibnamefont {Perinotti}},\ }\href {\doibase 10.1103/PhysRevResearch.6.033136} {\bibfield  {journal} {\bibinfo  {journal} {Phys. Rev. Res.}\ }\textbf {\bibinfo {volume} {6}},\ \bibinfo {pages} {033136} (\bibinfo {year} {2024})}\BibitemShut {NoStop}%
\bibitem [{\citenamefont {D'Ariano}\ and\ \citenamefont {Perinotti}(2014)}]{PhysRevA.90.062106}%
  \BibitemOpen
  \bibfield  {author} {\bibinfo {author} {\bibfnamefont {G.~M.}\ \bibnamefont {D'Ariano}}\ and\ \bibinfo {author} {\bibfnamefont {P.}~\bibnamefont {Perinotti}},\ }\href {\doibase 10.1103/PhysRevA.90.062106} {\bibfield  {journal} {\bibinfo  {journal} {Phys. Rev. A}\ }\textbf {\bibinfo {volume} {90}},\ \bibinfo {pages} {062106} (\bibinfo {year} {2014})}\BibitemShut {NoStop}%
\bibitem [{\citenamefont {D'Ariano}\ and\ \citenamefont {Perinotti}(2016)}]{DAriano:2016aa}%
  \BibitemOpen
  \bibfield  {author} {\bibinfo {author} {\bibfnamefont {G.~M.}\ \bibnamefont {D'Ariano}}\ and\ \bibinfo {author} {\bibfnamefont {P.}~\bibnamefont {Perinotti}},\ }\href {\doibase 10.1007/s11467-016-0616-z} {\bibfield  {journal} {\bibinfo  {journal} {Frontiers of Physics}\ }\textbf {\bibinfo {volume} {12}},\ \bibinfo {pages} {120301} (\bibinfo {year} {2016})}\BibitemShut {NoStop}%
\bibitem [{\citenamefont {D'Ariano}\ \emph {et~al.}(2017)\citenamefont {D'Ariano}, \citenamefont {Erba},\ and\ \citenamefont {Perinotti}}]{PhysRevA.96.062101}%
  \BibitemOpen
  \bibfield  {author} {\bibinfo {author} {\bibfnamefont {G.~M.}\ \bibnamefont {D'Ariano}}, \bibinfo {author} {\bibfnamefont {M.}~\bibnamefont {Erba}}, \ and\ \bibinfo {author} {\bibfnamefont {P.}~\bibnamefont {Perinotti}},\ }\href {\doibase 10.1103/PhysRevA.96.062101} {\bibfield  {journal} {\bibinfo  {journal} {Phys. Rev. A}\ }\textbf {\bibinfo {volume} {96}},\ \bibinfo {pages} {062101} (\bibinfo {year} {2017})}\BibitemShut {NoStop}%
\bibitem [{\citenamefont {Perinotti}(2020)}]{Perinotti_2020}%
  \BibitemOpen
  \bibfield  {author} {\bibinfo {author} {\bibfnamefont {P.}~\bibnamefont {Perinotti}},\ }\href {\doibase 10.4081/scie.2017.649} {\bibfield  {journal} {\bibinfo  {journal} {Istituto Lombardo - Accademia di Scienze e Lettere - Rendiconti di Scienze}\ } (\bibinfo {year} {2020}),\ 10.4081/scie.2017.649}\BibitemShut {NoStop}%
\bibitem [{\citenamefont {Arrighi}\ \emph {et~al.}(2014)\citenamefont {Arrighi}, \citenamefont {Nesme},\ and\ \citenamefont {Forets}}]{Arrighi_2014}%
  \BibitemOpen
  \bibfield  {author} {\bibinfo {author} {\bibfnamefont {P.}~\bibnamefont {Arrighi}}, \bibinfo {author} {\bibfnamefont {V.}~\bibnamefont {Nesme}}, \ and\ \bibinfo {author} {\bibfnamefont {M.}~\bibnamefont {Forets}},\ }\href {\doibase 10.1088/1751-8113/47/46/465302} {\bibfield  {journal} {\bibinfo  {journal} {Journal of Physics A: Mathematical and Theoretical}\ }\textbf {\bibinfo {volume} {47}},\ \bibinfo {pages} {465302} (\bibinfo {year} {2014})}\BibitemShut {NoStop}%
\bibitem [{\citenamefont {Eon}\ \emph {et~al.}(2023)\citenamefont {Eon}, \citenamefont {Molfetta}, \citenamefont {Magnifico},\ and\ \citenamefont {Arrighi}}]{Eon2023relativistic}%
  \BibitemOpen
  \bibfield  {author} {\bibinfo {author} {\bibfnamefont {N.}~\bibnamefont {Eon}}, \bibinfo {author} {\bibfnamefont {G.~D.}\ \bibnamefont {Molfetta}}, \bibinfo {author} {\bibfnamefont {G.}~\bibnamefont {Magnifico}}, \ and\ \bibinfo {author} {\bibfnamefont {P.}~\bibnamefont {Arrighi}},\ }\href {\doibase 10.22331/q-2023-11-08-1179} {\bibfield  {journal} {\bibinfo  {journal} {{Quantum}}\ }\textbf {\bibinfo {volume} {7}},\ \bibinfo {pages} {1179} (\bibinfo {year} {2023})}\BibitemShut {NoStop}%
\bibitem [{\citenamefont {Haah}(2021)}]{haah2021clifford}%
  \BibitemOpen
  \bibfield  {author} {\bibinfo {author} {\bibfnamefont {J.}~\bibnamefont {Haah}},\ }\href {\doibase 10.1063/5.0022185} {\bibfield  {journal} {\bibinfo  {journal} {Journal of Mathematical Physics}\ }\textbf {\bibinfo {volume} {62}},\ \bibinfo {pages} {092202} (\bibinfo {year} {2021})}\BibitemShut {NoStop}%
\bibitem [{\citenamefont {G{\"u}tschow}\ \emph {et~al.}(2010)\citenamefont {G{\"u}tschow}, \citenamefont {Uphoff}, \citenamefont {Werner},\ and\ \citenamefont {Zimbor{\'a}s}}]{10.1063/1.3278513}%
  \BibitemOpen
  \bibfield  {author} {\bibinfo {author} {\bibfnamefont {J.}~\bibnamefont {G{\"u}tschow}}, \bibinfo {author} {\bibfnamefont {S.}~\bibnamefont {Uphoff}}, \bibinfo {author} {\bibfnamefont {R.~F.}\ \bibnamefont {Werner}}, \ and\ \bibinfo {author} {\bibfnamefont {Z.}~\bibnamefont {Zimbor{\'a}s}},\ }\href {\doibase 10.1063/1.3278513} {\bibfield  {journal} {\bibinfo  {journal} {Journal of Mathematical Physics}\ }\textbf {\bibinfo {volume} {51}},\ \bibinfo {pages} {015203} (\bibinfo {year} {2010})}\BibitemShut {NoStop}%
\bibitem [{\citenamefont {Haah}(2024)}]{haah2022topological}%
  \BibitemOpen
  \bibfield  {author} {\bibinfo {author} {\bibfnamefont {J.}~\bibnamefont {Haah}},\ }\href {https://arxiv.org/abs/2205.09141} {\  (\bibinfo {year} {2024})},\ \Eprint {http://arxiv.org/abs/2205.09141} {arXiv:2205.09141 [math-ph]} \BibitemShut {NoStop}%
\bibitem [{\citenamefont {Englbrecht}\ and\ \citenamefont {Kraus}(2020)}]{PhysRevA.101.062302}%
  \BibitemOpen
  \bibfield  {author} {\bibinfo {author} {\bibfnamefont {M.}~\bibnamefont {Englbrecht}}\ and\ \bibinfo {author} {\bibfnamefont {B.}~\bibnamefont {Kraus}},\ }\href {\doibase 10.1103/PhysRevA.101.062302} {\bibfield  {journal} {\bibinfo  {journal} {Phys. Rev. A}\ }\textbf {\bibinfo {volume} {101}},\ \bibinfo {pages} {062302} (\bibinfo {year} {2020})}\BibitemShut {NoStop}%
\bibitem [{\citenamefont {Gross}\ \emph {et~al.}(2012)\citenamefont {Gross}, \citenamefont {Nesme}, \citenamefont {Vogts},\ and\ \citenamefont {Werner}}]{GNVW}%
  \BibitemOpen
  \bibfield  {author} {\bibinfo {author} {\bibfnamefont {D.}~\bibnamefont {Gross}}, \bibinfo {author} {\bibfnamefont {V.}~\bibnamefont {Nesme}}, \bibinfo {author} {\bibfnamefont {H.}~\bibnamefont {Vogts}}, \ and\ \bibinfo {author} {\bibfnamefont {R.~F.}\ \bibnamefont {Werner}},\ }\href {\doibase https://doi.org/10.1007/s00220-012-1423-1} {\bibfield  {journal} {\bibinfo  {journal} {Communications in Mathematical Physics}\ }\textbf {\bibinfo {volume} {310}},\ \bibinfo {pages} {419} (\bibinfo {year} {2012})}\BibitemShut {NoStop}%
\bibitem [{\citenamefont {Vogts}(2009)}]{Vogts2009DiscreteTQ}%
  \BibitemOpen
  \bibfield  {author} {\bibinfo {author} {\bibfnamefont {H.}~\bibnamefont {Vogts}},\ }\emph {\bibinfo {title} {Discrete Time Quantum Lattice Systems}},\ \href {https://api.semanticscholar.org/CorpusID:125606673} {\bibinfo {type} {{PhD} dissertation}},\ \bibinfo  {school} {Technische Universit\"at Braunschweig} (\bibinfo {year} {2009})\BibitemShut {NoStop}%
\bibitem [{\citenamefont {Freedman}\ and\ \citenamefont {Hastings}(2020)}]{freedman2020classification}%
  \BibitemOpen
  \bibfield  {author} {\bibinfo {author} {\bibfnamefont {M.}~\bibnamefont {Freedman}}\ and\ \bibinfo {author} {\bibfnamefont {M.~B.}\ \bibnamefont {Hastings}},\ }\href {\doibase https://doi.org/10.1007/s00220-020-03735-y} {\bibfield  {journal} {\bibinfo  {journal} {Communications in Mathematical Physics}\ }\textbf {\bibinfo {volume} {376}},\ \bibinfo {pages} {1171} (\bibinfo {year} {2020})}\BibitemShut {NoStop}%
\bibitem [{\citenamefont {Haah}\ \emph {et~al.}(2023)\citenamefont {Haah}, \citenamefont {Fidkowski},\ and\ \citenamefont {Hastings}}]{haah2023nontrivial}%
  \BibitemOpen
  \bibfield  {author} {\bibinfo {author} {\bibfnamefont {J.}~\bibnamefont {Haah}}, \bibinfo {author} {\bibfnamefont {L.}~\bibnamefont {Fidkowski}}, \ and\ \bibinfo {author} {\bibfnamefont {M.~B.}\ \bibnamefont {Hastings}},\ }\href {\doibase https://doi.org/10.1007/s00220-022-04528-1} {\bibfield  {journal} {\bibinfo  {journal} {Communications in Mathematical Physics}\ }\textbf {\bibinfo {volume} {398}},\ \bibinfo {pages} {469} (\bibinfo {year} {2023})}\BibitemShut {NoStop}%
\bibitem [{\citenamefont {Stephen}\ \emph {et~al.}(2019)\citenamefont {Stephen}, \citenamefont {Nautrup}, \citenamefont {Bermejo-Vega}, \citenamefont {Eisert},\ and\ \citenamefont {Raussendorf}}]{stephen2019subsystem}%
  \BibitemOpen
  \bibfield  {author} {\bibinfo {author} {\bibfnamefont {D.~T.}\ \bibnamefont {Stephen}}, \bibinfo {author} {\bibfnamefont {H.~P.}\ \bibnamefont {Nautrup}}, \bibinfo {author} {\bibfnamefont {J.}~\bibnamefont {Bermejo-Vega}}, \bibinfo {author} {\bibfnamefont {J.}~\bibnamefont {Eisert}}, \ and\ \bibinfo {author} {\bibfnamefont {R.}~\bibnamefont {Raussendorf}},\ }\href {\doibase https://doi.org/10.22331/q-2019-05-20-142} {\bibfield  {journal} {\bibinfo  {journal} {Quantum}\ }\textbf {\bibinfo {volume} {3}},\ \bibinfo {pages} {142} (\bibinfo {year} {2019})}\BibitemShut {NoStop}%
\bibitem [{\citenamefont {Po}\ \emph {et~al.}(2016)\citenamefont {Po}, \citenamefont {Fidkowski}, \citenamefont {Morimoto}, \citenamefont {Potter},\ and\ \citenamefont {Vishwanath}}]{po2016chiral}%
  \BibitemOpen
  \bibfield  {author} {\bibinfo {author} {\bibfnamefont {H.~C.}\ \bibnamefont {Po}}, \bibinfo {author} {\bibfnamefont {L.}~\bibnamefont {Fidkowski}}, \bibinfo {author} {\bibfnamefont {T.}~\bibnamefont {Morimoto}}, \bibinfo {author} {\bibfnamefont {A.~C.}\ \bibnamefont {Potter}}, \ and\ \bibinfo {author} {\bibfnamefont {A.}~\bibnamefont {Vishwanath}},\ }\href {\doibase https://doi.org/10.1103/PhysRevX.6.041070} {\bibfield  {journal} {\bibinfo  {journal} {Physical Review X}\ }\textbf {\bibinfo {volume} {6}},\ \bibinfo {pages} {041070} (\bibinfo {year} {2016})}\BibitemShut {NoStop}%
\bibitem [{\citenamefont {Fidkowski}\ \emph {et~al.}(2019)\citenamefont {Fidkowski}, \citenamefont {Po}, \citenamefont {Potter},\ and\ \citenamefont {Vishwanath}}]{fidkowski2019interacting}%
  \BibitemOpen
  \bibfield  {author} {\bibinfo {author} {\bibfnamefont {L.}~\bibnamefont {Fidkowski}}, \bibinfo {author} {\bibfnamefont {H.~C.}\ \bibnamefont {Po}}, \bibinfo {author} {\bibfnamefont {A.~C.}\ \bibnamefont {Potter}}, \ and\ \bibinfo {author} {\bibfnamefont {A.}~\bibnamefont {Vishwanath}},\ }\href {\doibase https://doi.org/10.1103/PhysRevB.99.085115} {\bibfield  {journal} {\bibinfo  {journal} {Physical Review B}\ }\textbf {\bibinfo {volume} {99}},\ \bibinfo {pages} {085115} (\bibinfo {year} {2019})}\BibitemShut {NoStop}%
\bibitem [{\citenamefont {Kitaev}\ and\ \citenamefont {Preskill}(2006)}]{kitaev2006topological}%
  \BibitemOpen
  \bibfield  {author} {\bibinfo {author} {\bibfnamefont {A.}~\bibnamefont {Kitaev}}\ and\ \bibinfo {author} {\bibfnamefont {J.}~\bibnamefont {Preskill}},\ }\href {\doibase https://doi.org/10.1103/PhysRevLett.96.110404} {\bibfield  {journal} {\bibinfo  {journal} {Physical Review Letters}\ }\textbf {\bibinfo {volume} {96}},\ \bibinfo {pages} {110404} (\bibinfo {year} {2006})}\BibitemShut {NoStop}%
\bibitem [{\citenamefont {Gong}\ \emph {et~al.}(2021)\citenamefont {Gong}, \citenamefont {Piroli},\ and\ \citenamefont {Cirac}}]{gong2021topological}%
  \BibitemOpen
  \bibfield  {author} {\bibinfo {author} {\bibfnamefont {Z.}~\bibnamefont {Gong}}, \bibinfo {author} {\bibfnamefont {L.}~\bibnamefont {Piroli}}, \ and\ \bibinfo {author} {\bibfnamefont {J.~I.}\ \bibnamefont {Cirac}},\ }\href {\doibase https://doi.org/10.1103/PhysRevLett.126.160601} {\bibfield  {journal} {\bibinfo  {journal} {Physical Review Letters}\ }\textbf {\bibinfo {volume} {126}},\ \bibinfo {pages} {160601} (\bibinfo {year} {2021})}\BibitemShut {NoStop}%
\bibitem [{\citenamefont {Ranard}\ \emph {et~al.}(2022)\citenamefont {Ranard}, \citenamefont {Walter},\ and\ \citenamefont {Witteveen}}]{ranard2022converse}%
  \BibitemOpen
  \bibfield  {author} {\bibinfo {author} {\bibfnamefont {D.}~\bibnamefont {Ranard}}, \bibinfo {author} {\bibfnamefont {M.}~\bibnamefont {Walter}}, \ and\ \bibinfo {author} {\bibfnamefont {F.}~\bibnamefont {Witteveen}},\ }in\ \href {\doibase https://doi.org/10.1007/s00023-022-01193-x} {\emph {\bibinfo {booktitle} {Annales Henri Poincar{\'e}}}},\ Vol.~\bibinfo {volume} {23}\ (\bibinfo {organization} {Springer},\ \bibinfo {year} {2022})\ pp.\ \bibinfo {pages} {3905--3979}\BibitemShut {NoStop}%
\bibitem [{\citenamefont {Horodecki}\ \emph {et~al.}(2009)\citenamefont {Horodecki}, \citenamefont {Horodecki}, \citenamefont {Horodecki},\ and\ \citenamefont {Horodecki}}]{horodecki2009quantum}%
  \BibitemOpen
  \bibfield  {author} {\bibinfo {author} {\bibfnamefont {R.}~\bibnamefont {Horodecki}}, \bibinfo {author} {\bibfnamefont {P.}~\bibnamefont {Horodecki}}, \bibinfo {author} {\bibfnamefont {M.}~\bibnamefont {Horodecki}}, \ and\ \bibinfo {author} {\bibfnamefont {K.}~\bibnamefont {Horodecki}},\ }\href {\doibase https://doi.org/10.1103/RevModPhys.81.865} {\bibfield  {journal} {\bibinfo  {journal} {Reviews of modern physics}\ }\textbf {\bibinfo {volume} {81}},\ \bibinfo {pages} {865} (\bibinfo {year} {2009})}\BibitemShut {NoStop}%
\bibitem [{\citenamefont {Wootters}\ and\ \citenamefont {Zurek}(1982)}]{Wootters:1982aa}%
  \BibitemOpen
  \bibfield  {author} {\bibinfo {author} {\bibfnamefont {W.~K.}\ \bibnamefont {Wootters}}\ and\ \bibinfo {author} {\bibfnamefont {W.~H.}\ \bibnamefont {Zurek}},\ }\href {\doibase 10.1038/299802a0} {\bibfield  {journal} {\bibinfo  {journal} {Nature}\ }\textbf {\bibinfo {volume} {299}},\ \bibinfo {pages} {802} (\bibinfo {year} {1982})}\BibitemShut {NoStop}%
\bibitem [{Note1()}]{Note1}%
  \BibitemOpen
  \bibinfo {note} {We remind the reader that a C*-algebra $\protect \mathcal A$ is an algebra of operators over some complex Hilbert space that is closed under adjoint and complete (every Cauchy sequence in the operator norm converges within the algebra itself), with $\|A^\protect \dag A\|=\|A\|^2$ and $\|AB\|\leq \|A\|\protect \,\|B\|$ for all elements $A,B\in \protect \mathcal A$.}\BibitemShut {Stop}%
\bibitem [{\citenamefont {Mac~Lane}(1998)}]{zbMATH01216133}%
  \BibitemOpen
  \bibfield  {author} {\bibinfo {author} {\bibfnamefont {S.}~\bibnamefont {Mac~Lane}},\ }\href {\doibase https://doi.org/10.1007/978-1-4757-4721-8} {\emph {\bibinfo {title} {Categories for the working mathematician.}}},\ \bibinfo {series} {Grad. Texts Math.}, Vol.~\bibinfo {volume} {5}\ (\bibinfo  {publisher} {New York, NY: Springer},\ \bibinfo {year} {1998})\BibitemShut {NoStop}%
\bibitem [{\citenamefont {Ruelle}(1999)}]{doi:10.1142/4090}%
  \BibitemOpen
  \bibfield  {author} {\bibinfo {author} {\bibfnamefont {D.}~\bibnamefont {Ruelle}},\ }\href {\doibase 10.1142/4090} {\emph {\bibinfo {title} {Statistical Mechanics}}}\ (\bibinfo  {publisher} {Co-published with Imperial College Press},\ \bibinfo {year} {1999})\BibitemShut {NoStop}%
\bibitem [{\citenamefont {Bratteli}\ and\ \citenamefont {Robinson}(1979)}]{brattelirobinson}%
  \BibitemOpen
  \bibfield  {author} {\bibinfo {author} {\bibfnamefont {O.}~\bibnamefont {Bratteli}}\ and\ \bibinfo {author} {\bibfnamefont {D.~W.}\ \bibnamefont {Robinson}},\ }\href {\doibase https://doi.org/10.1007/978-3-662-02520-8} {\emph {\bibinfo {title} {Operator Algebras and Quantum Statistical Mechanics Volume 1: C*-and W*-Algebras. Symmetry Groups. Decomposition of States}}}\ (\bibinfo  {publisher} {Springer},\ \bibinfo {year} {1979})\BibitemShut {NoStop}%
\bibitem [{Note2()}]{Note2}%
  \BibitemOpen
  \bibinfo {note} {We used the additive notation which is suitable only if translations on the lattice are abelian. This is the case for the present Letter, even though more general situations can be conceived~\cite {ARRIGHI_MARTIEL_NESME_2018,DAriano:2016aa,Perinotti_2020}}\BibitemShut {NoStop}%
\bibitem [{Note3()}]{Note3}%
  \BibitemOpen
  \bibinfo {note} {We say that two algebras commute if all their operators commute.}\BibitemShut {Stop}%
\bibitem [{Note4()}]{Note4}%
  \BibitemOpen
  \bibinfo {note} {It has been proved that allowing for an ancillary copy of the system \textcolor {black}{one can \protect \emph {simulate} \protect \emph {any} QCA via a finite-depth quantum circuit~\cite {ARRIGHI2011372,arrighi2012intrinsically,arrighi2019overview}. Specifically, given any QCA $\alpha $ on $\protect \mathcal {A}_{\protect \mathbb {Z}^s}$, the QCA $\alpha \otimes \alpha ^{-1}: \protect \mathcal {A}_{\protect \mathbb {Z}^s} \otimes \protect \mathcal {A}_{\protect \mathbb {Z}^s} \to \protect \mathcal {A}_{\protect \mathbb {Z}^s} \otimes \protect \mathcal {A}_{\protect \mathbb {Z}^s} $, which acts as $ A \otimes B \DOTSB \mapstochar \rightarrow \alpha (A) \otimes \alpha ^{-1}(B)$, is realizable as an FDQC. This circuit simulates the automaton $\alpha $ on the first copy of the system.}}\BibitemShut {Stop}%
\bibitem [{\citenamefont {Zanardi}(2000)}]{PhysRevA.63.012301}%
  \BibitemOpen
  \bibfield  {author} {\bibinfo {author} {\bibfnamefont {P.}~\bibnamefont {Zanardi}},\ }\href {\doibase 10.1103/PhysRevA.63.012301} {\bibfield  {journal} {\bibinfo  {journal} {Phys. Rev. A}\ }\textbf {\bibinfo {volume} {63}},\ \bibinfo {pages} {012301} (\bibinfo {year} {2000})}\BibitemShut {NoStop}%
\bibitem [{sup()}]{supp}%
  \BibitemOpen
  \href@noop {} {}\bibinfo {note} {See Supplemental Material below for detailed discussion.}\BibitemShut {Stop}%
\bibitem [{Note5()}]{Note5}%
  \BibitemOpen
  \bibinfo {note} {An algebra consists of all finite linear combinations over $\protect \mathbb {C}$ of finite products of its generators.}\BibitemShut {Stop}%
\bibitem [{Note6()}]{Note6}%
  \BibitemOpen
  \bibinfo {note} {Here the result is stated for translation-invariant QCA, but it can be proved more generally under the hypothesis of locality~\cite {GNVW}, \textcolor {black}{and for multiple $1$ dimensional chains linked together~\cite {freedman2020classification}.}}\BibitemShut {Stop}%
\bibitem [{Note7()}]{Note7}%
  \BibitemOpen
  \bibinfo {note} {An algebra $\protect \mathcal {A}$ is called simple if its closed two-sided ideals are only $\{0\}$ and $\protect \mathcal {A}$. A matrix algebra with a trivial center is simple. A quasi-local algebra of simple algebras is simple \cite {brattelirobinson}.}\BibitemShut {Stop}%
\bibitem [{\citenamefont {Van~den Nest}\ \emph {et~al.}(2007)\citenamefont {Van~den Nest}, \citenamefont {D{\"u}r}, \citenamefont {Miyake},\ and\ \citenamefont {Briegel}}]{van2007fundamentals}%
  \BibitemOpen
  \bibfield  {author} {\bibinfo {author} {\bibfnamefont {M.}~\bibnamefont {Van~den Nest}}, \bibinfo {author} {\bibfnamefont {W.}~\bibnamefont {D{\"u}r}}, \bibinfo {author} {\bibfnamefont {A.}~\bibnamefont {Miyake}}, \ and\ \bibinfo {author} {\bibfnamefont {H.~J.}\ \bibnamefont {Briegel}},\ }\href {\doibase https://doi.org/10.1088/1367-2630/9/6/204} {\bibfield  {journal} {\bibinfo  {journal} {New Journal of Physics}\ }\textbf {\bibinfo {volume} {9}},\ \bibinfo {pages} {204} (\bibinfo {year} {2007})}\BibitemShut {NoStop}%
\bibitem [{\citenamefont {Gross}\ \emph {et~al.}(2007)\citenamefont {Gross}, \citenamefont {Eisert}, \citenamefont {Schuch},\ and\ \citenamefont {Perez-Garcia}}]{gross2007measurement}%
  \BibitemOpen
  \bibfield  {author} {\bibinfo {author} {\bibfnamefont {D.}~\bibnamefont {Gross}}, \bibinfo {author} {\bibfnamefont {J.}~\bibnamefont {Eisert}}, \bibinfo {author} {\bibfnamefont {N.}~\bibnamefont {Schuch}}, \ and\ \bibinfo {author} {\bibfnamefont {D.}~\bibnamefont {Perez-Garcia}},\ }\href {\doibase https://doi.org/10.1103/PhysRevA.76.052315} {\bibfield  {journal} {\bibinfo  {journal} {Physical Review A}\ }\textbf {\bibinfo {volume} {76}},\ \bibinfo {pages} {052315} (\bibinfo {year} {2007})}\BibitemShut {NoStop}%
\bibitem [{\citenamefont {Poulsen~Nautrup}\ and\ \citenamefont {Briegel}(2024)}]{PhysRevA.110.062617}%
  \BibitemOpen
  \bibfield  {author} {\bibinfo {author} {\bibfnamefont {H.}~\bibnamefont {Poulsen~Nautrup}}\ and\ \bibinfo {author} {\bibfnamefont {H.~J.}\ \bibnamefont {Briegel}},\ }\href {\doibase 10.1103/PhysRevA.110.062617} {\bibfield  {journal} {\bibinfo  {journal} {Phys. Rev. A}\ }\textbf {\bibinfo {volume} {110}},\ \bibinfo {pages} {062617} (\bibinfo {year} {2024})}\BibitemShut {NoStop}%
\bibitem [{Note8()}]{Note8}%
  \BibitemOpen
  \bibinfo {note} {The algorithm was developed using the Qiskit library (version 0.45.3), and its performances match those reported in \cite {suzuki2021qulacs}. When employing JAX and Pennylane libraries (versions 0.4.24 and 0.34.0), we encounter a significant overhead as the number of qubits and parameters grows, resulting in a less efficient implementation.}\BibitemShut {Stop}%
\bibitem [{\citenamefont {Arrighi}\ \emph {et~al.}(2018)\citenamefont {Arrighi}, \citenamefont {Martiel},\ and\ \citenamefont {Nesme}}]{ARRIGHI_MARTIEL_NESME_2018}%
  \BibitemOpen
  \bibfield  {author} {\bibinfo {author} {\bibfnamefont {P.}~\bibnamefont {Arrighi}}, \bibinfo {author} {\bibfnamefont {S.}~\bibnamefont {Martiel}}, \ and\ \bibinfo {author} {\bibfnamefont {V.}~\bibnamefont {Nesme}},\ }\href {\doibase 10.1017/S0960129517000044} {\bibfield  {journal} {\bibinfo  {journal} {Mathematical Structures in Computer Science}\ }\textbf {\bibinfo {volume} {28}},\ \bibinfo {pages} {340} (\bibinfo {year} {2018})}\BibitemShut {NoStop}%
\bibitem [{\citenamefont {Arrighi}\ \emph {et~al.}(2011)\citenamefont {Arrighi}, \citenamefont {Nesme},\ and\ \citenamefont {Werner}}]{ARRIGHI2011372}%
  \BibitemOpen
  \bibfield  {author} {\bibinfo {author} {\bibfnamefont {P.}~\bibnamefont {Arrighi}}, \bibinfo {author} {\bibfnamefont {V.}~\bibnamefont {Nesme}}, \ and\ \bibinfo {author} {\bibfnamefont {R.}~\bibnamefont {Werner}},\ }\href {\doibase https://doi.org/10.1016/j.jcss.2010.05.004} {\bibfield  {journal} {\bibinfo  {journal} {Journal of Computer and System Sciences}\ }\textbf {\bibinfo {volume} {77}},\ \bibinfo {pages} {372} (\bibinfo {year} {2011})},\ \bibinfo {note} {adaptivity in Heterogeneous Environments}\BibitemShut {NoStop}%
\bibitem [{\citenamefont {Arrighi}\ and\ \citenamefont {Grattage}(2012)}]{arrighi2012intrinsically}%
  \BibitemOpen
  \bibfield  {author} {\bibinfo {author} {\bibfnamefont {P.}~\bibnamefont {Arrighi}}\ and\ \bibinfo {author} {\bibfnamefont {J.}~\bibnamefont {Grattage}},\ }\href {https://www.sciencedirect.com/science/article/pii/S002200001100153X} {\bibfield  {journal} {\bibinfo  {journal} {Journal of Computer and System Sciences}\ }\textbf {\bibinfo {volume} {78}},\ \bibinfo {pages} {1883} (\bibinfo {year} {2012})}\BibitemShut {NoStop}%
\bibitem [{\citenamefont {Suzuki}\ \emph {et~al.}(2021)\citenamefont {Suzuki}, \citenamefont {Kawase}, \citenamefont {Masumura}, \citenamefont {Hiraga}, \citenamefont {Nakadai}, \citenamefont {Chen}, \citenamefont {Nakanishi}, \citenamefont {Mitarai}, \citenamefont {Imai}, \citenamefont {Tamiya}, \citenamefont {Yamamoto}, \citenamefont {Yan}, \citenamefont {Kawakubo}, \citenamefont {Nakagawa}, \citenamefont {Ibe}, \citenamefont {Zhang}, \citenamefont {Yamashita}, \citenamefont {Yoshimura}, \citenamefont {Hayashi},\ and\ \citenamefont {Fujii}}]{suzuki2021qulacs}%
  \BibitemOpen
  \bibfield  {author} {\bibinfo {author} {\bibfnamefont {Y.}~\bibnamefont {Suzuki}}, \bibinfo {author} {\bibfnamefont {Y.}~\bibnamefont {Kawase}}, \bibinfo {author} {\bibfnamefont {Y.}~\bibnamefont {Masumura}}, \bibinfo {author} {\bibfnamefont {Y.}~\bibnamefont {Hiraga}}, \bibinfo {author} {\bibfnamefont {M.}~\bibnamefont {Nakadai}}, \bibinfo {author} {\bibfnamefont {J.}~\bibnamefont {Chen}}, \bibinfo {author} {\bibfnamefont {K.~M.}\ \bibnamefont {Nakanishi}}, \bibinfo {author} {\bibfnamefont {K.}~\bibnamefont {Mitarai}}, \bibinfo {author} {\bibfnamefont {R.}~\bibnamefont {Imai}}, \bibinfo {author} {\bibfnamefont {S.}~\bibnamefont {Tamiya}}, \bibinfo {author} {\bibfnamefont {T.}~\bibnamefont {Yamamoto}}, \bibinfo {author} {\bibfnamefont {T.}~\bibnamefont {Yan}}, \bibinfo {author} {\bibfnamefont {T.}~\bibnamefont {Kawakubo}}, \bibinfo {author} {\bibfnamefont {Y.~O.}\ \bibnamefont {Nakagawa}}, \bibinfo {author} {\bibfnamefont {Y.}~\bibnamefont {Ibe}}, \bibinfo {author} {\bibfnamefont {Y.}~\bibnamefont {Zhang}},
  \bibinfo {author} {\bibfnamefont {H.}~\bibnamefont {Yamashita}}, \bibinfo {author} {\bibfnamefont {H.}~\bibnamefont {Yoshimura}}, \bibinfo {author} {\bibfnamefont {A.}~\bibnamefont {Hayashi}}, \ and\ \bibinfo {author} {\bibfnamefont {K.}~\bibnamefont {Fujii}},\ }\href {\doibase 10.22331/q-2021-10-06-559} {\bibfield  {journal} {\bibinfo  {journal} {{Quantum}}\ }\textbf {\bibinfo {volume} {5}},\ \bibinfo {pages} {559} (\bibinfo {year} {2021})}\BibitemShut {NoStop}%
\end{thebibliography}%

\appendix
\newpage
\section{SUPPLEMENTAL MATERIAL}
This Supplemental Material provides complete proofs and comprehensive details of the results presented in the body of the manuscript. Section A reviews the essential tools, i.e.~the support algebras, employed in the two subsequent sections. Section B is centered on the proof of the update rules classification Theorem~1. Section C delves into the proof of index-based classification Theorem~2.


%\tableofcontents

%-----------------------------------------------------------------
%----------------------------- Contents
%-----------------------------------------------------------------
For the sake of rigor, some assertions made in the body of the manuscript are presented here in the form of lemmas. In what follows, we will use the following shorthand notation
\begin{align*}
 \xi_{\bm{x}} \coloneqq  \left(\bm{x}+\cN\right) \cap \cN.
\end{align*}

\section{Section A: Preliminary Results}\label{sec:app1}
Let $\alpha: \cA_{\Z^s} \to \cA_{\Z^s}$ be a QCA as per Definition~1.
\begin{lemma}\cite{schumacher2004reversible}
\mbox{}
\begin{itemize} \label{lem:iff_werner}
    \item The \emph{global update rule} $\alpha$ is uniquely determined by the \emph{local update rule} $ \alpha_0: \mathcal{A}_{\bm{0}}  \to \mathcal{A}_{\bm{0}+\mathcal{N}}\,$.
    \item A $*$-homomorphism $ \alpha_0: \mathcal{A}_{\bm{0}}  \to \mathcal{A}_{\bm{0}+\mathcal{N}}\,$ is the local update rule of a QCA $\alpha$ if and only if for all $\bm{0} \neq \bm{x} \in \mathbb{Z}^s$ such that $ \xi_{\bm{x}} \neq \emptyset\,$, the algebras $\alpha(\cA_{\bm{0}})$ and $\tau^{\bm{x}}(\alpha_0(\cA_{\bm{0}}))$ commute elementwise.
\end{itemize}
\end{lemma}

\begin{definition}
    Let  $\Lambda \in \left[ \Z^s\right]$ and $\{ \Omega_i\}$ be a partition of $\Lambda +~ \cN$. Given $\alpha(\cA_{\Lambda}) \se \cA_{\Lambda + \cN} = \bigotimes_i \cA_{\Omega_i} $, the \emph{support algebras} $\cS_{\Omega_i}^{\Lambda}$
    of $\alpha(\cA_{\Lambda})$ on $ \Omega_i$ are the smallest 
    C*-~subalgebras of $\cA_{\Omega_i}$ 
    such that
    \[ \alpha(\cA_{\Lambda}) \se \bigotimes_i \cS_{\Omega_i}^\Lambda   \se \bigotimes_i \cA_{\Omega_i}\,. \]
\end{definition}
\noindent
    In other terms, writing $O \in \alpha(\cA_{\Lambda})$ as $O=~ \sum_{\bm k}\bigotimes_{i} O_{\Omega_i}^{(k_i)}$, where $\bm k\coloneqq(k_1,k_2,\ldots)$ and for fixed $i$ the operators $O^{(k_i)}_{\Omega_i}$ are linearly independent, the support algebra $\cS_{\Omega_i}^\Lambda$ of $\alpha(\cA_{\Lambda})$ on $\Omega_i$ is the C*-algebra generated by all such $O_{\Omega_i}^{(k_i)}$.
\begin{lemma}\cite{schumacher2004reversible, Vogts2009DiscreteTQ} \label{lem:overlaplemma}
\leavevmode
\begin{enumerate}
\item\label{itlem1} Let $\{ \Gamma_j\} $ be a refinement of a partition $ \{ \Omega_i\}  $ of $\Lambda +~\cN$, then  $\bigotimes_i \cS_{\Omega_i}^\Lambda \se \bigotimes_j \cS_{\Gamma_j}^\Lambda$.
\item\label{itlem2} 
If $\{\Omega_1, \Omega_2\}$ is a partition of $\Omega$ and $\cS_{\Omega_1}^\Lambda \cong \cI_{\Omega_1}$, then  $\cS_{\Omega}^\Lambda \cong \cI_{\Omega_1} \otimes \cS_{\Omega_2}^\Lambda $. 
\item\label{itlem3} 
If $\Lambda_1 \cap \Lambda_2 = \emptyset$ and $\Omega=(\Lambda_1 + \cN) \cap (\Lambda_2 + \cN)$, 
then $\left[\cS_{\Omega}^{\Lambda_1},\,\cS_{\Omega}^{\Lambda_2} \right]=0$.

\item\label{itlem4} By translation invariance we have $\cS_{\Omega}^{ \Lambda+\bm{x} } \cong \cS_{\Omega-\bm{x}}^\Lambda$.
 \item\label{itlem5} $\cS_{\Omega}^\Lambda = \bigvee_{\bm{x} \in \Lambda} \cS_{ \Omega}^{\bm{x}}$ where the symbol $\cA\lor\cB$ denotes the smallest algebra generated by $\cA\cup\cB$.
\end{enumerate}
\end{lemma}
\section{Section B: Classification of Update Rules }\label{sec:app2}
Let $\alpha: \cA_{\Z^s} \to \cA_{\Z^s}$ be a QCA of qubit systems $\cA_{\bm x} \cong \cM_2$  with a von Neumann neighborhood scheme $\mathcal{N}\subseteq\{\bm{0}\} \cup  \{ \pm \bm{e}_i \}_{i=1}^s$, where $\{ \bm{e}_i \}_{i=1}^s$ is the canonical basis of $\Z^s$.
Let us begin with two overarching observations. First, it is imperative that at least one of the support algebras $\cS_{\bm{x}}^{\bm{0}}\,$, with $\bm x \in \cN$, is non-abelian since $\alpha_0$ is an automorphism of a non-abelian algebra. Second, the non-trivial C*-subalgebras of $\cM_2$ are all isomorphic, and they consist in the abelian algebras of simultaneously diagonalizable matrices
\begin{equation*}
    \cD(\bm{n}) \coloneqq \left\{ a\,I+b\,\bm{n} \cdot \bm{\sigma}  \; \mid \;  a,b \in \C \;;\; \bm{n} \in \mathbb{R}^3\,,\; \norm{\bm{n}}_2=1 \right\}\,,
\end{equation*} 
where $\bm{\sigma}=(\sigma^1, \sigma^2, \sigma^3)$ denote the Pauli matrices.

\begin{lemma}\label{lem:confsuppalg}
      Let $\alpha$ be a QCA with von Neumann neighborhood scheme $\cN$ of a hypercubic lattice $\mathbb{Z}^{s}$ of qubits. Then 
     \begin{enumerate}
         \item[i.] either $ \alpha(\cA_{\bm 0}) \se \cA_{\bm{x}} \otimes \cI_{\cN \setminus \bm{x}}$, with $\bm{x}\in\cN\,$,
         \item[ii.] or $\alpha(\cA_{\bm 0}) \se \cA_{\bm{0}} \otimes \bigotimes_{i=1}^s(\cD(\bm{n}_i)_{-{\bm e}_i}{\otimes}\cD(\bm{n}_i)_{{\bm e}_i})\,$,
     \end{enumerate}
     where the r.h.s is $\bigotimes_{\bm x \in \cN} \cS^{\bm 0}_{\bm x}\,$.
 \end{lemma}

 \begin{proof}
     By item~\textit{3.} of 
     Lemma \ref{lem:overlaplemma} we have
     \begin{equation}\label{supp alg comm}
         \left[\cS_{\xi_{\bm{x}}}^{\bm{0}},\, \cS_{\xi_{\bm{x}}}^{\bm{x}}\right]=0 
     \end{equation} 
     for each pair $(\bm{x},\,\xi_{\bm{x}})$ among the following 
     \begin{enumerate}
         \item\label{comm1} $\left( 2 \bm{e}_i\,,\, \{\bm{e}_i\} \right)\,,$
         \item\label{comm2} $\left( \bm{e}_i\pm\bm{e}_j\,,\,\{\bm{e}_i, \pm \bm{e}_j\}\right)$ with $i \neq j\,,$
         \item\label{comm3} $\left( \bm{e}_i\,,\,\{ \bm{0},\bm{e}_i \} \right)$\,,
     \end{enumerate}
     with the overlap sets $\xi_{\bm{x}}$ as in Lemma~\ref{lem:iff_werner}, namely $ \xi_{\bm{x}} = (\bm{0} + \cN) \cap \left(\bm{x}+\cN\right)\neq \emptyset$ (FIG.~\ref{fig:overlaps}).
     \begin{figure}[h]
    \centering
    \includegraphics[width=8.5cm]{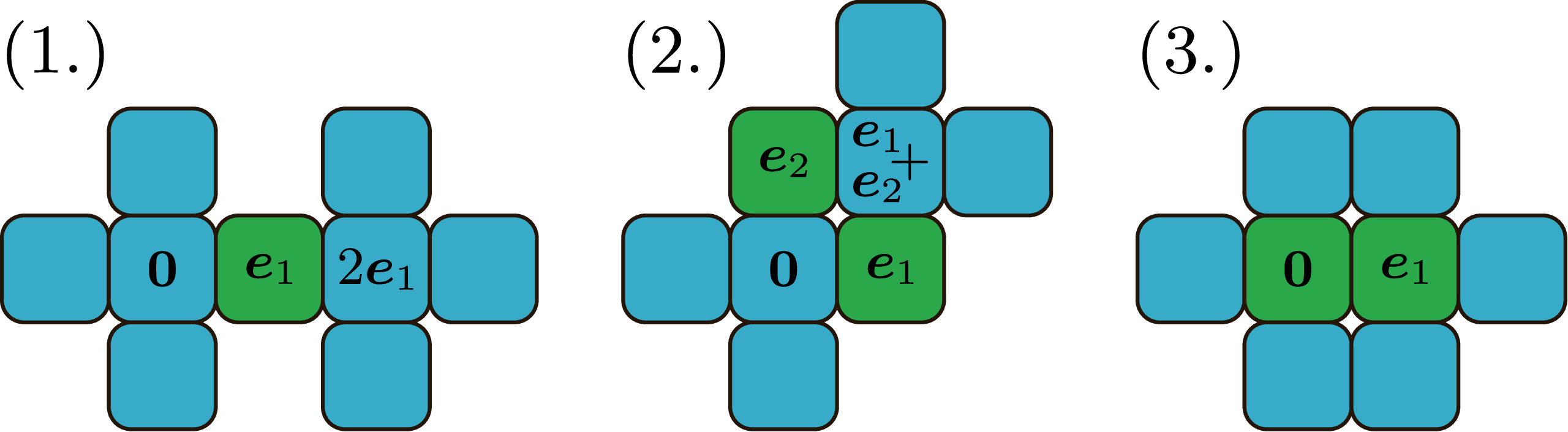}
    \caption{Instances of overlapping neighborhoods for the relevant cases of $(\bm{x},\,\xi_{\bm{x}})$ listed above (\ref{supp alg comm}). The neighborhoods of $\bm{0}$ and $\bm{x}$ are depicted in blue while their overlap $\xi_{\bm x}$ in green.}
    \label{fig:overlaps}
\end{figure}
     First of all, we mark that by translation invariance
     we have $ \cS_{\xi_{\bm{x}}}^{\bm{x}} \cong \cS_{-\xi_{\bm{x}}}^{\bm{0}}$. Thus, commutation relations for case~\ref{comm1} impose the support algebras $\cS_{\pm\bm{e}_i}^{\bm{0}}$ located at opposite sites of $\cN$ to be isomorphic to commuting algebras 
     \[(\cS_{-\bm{e}_i}^{\bm{0}}, \cS_{+\bm{e}_i}^{\bm{0}}) \cong  \left\{
     \begin{aligned}
     &(\cI, \cI),\\ 
     &(\cD(\bm{n}_i),\cD(\bm{n}_i)), \\
     &(\cI, \cM_2),\\
     &(\cI,\cD(\bm{n}_i)).
     \end{aligned}\right.\]
     (and permutations).\\
     Let us examine these four possible configurations.\par
     For a given $i$ let $(\cS_{-\bm{e}_i}^{\bm{0}},\cS_{\bm{e}_i}^{\bm{0}}) \cong (\cI,\cM_2)$, then the conditions from case~\ref{comm2} impose $ \cS_{\pm\bm{e_{j}}}^{\bm{0}} \cong \cI\; \forall\,j\neq i$ and those from case~\ref{comm3} fix $\cS_{\bm{0}}^{\bm{0}} \cong \cI$. That is, if one of the support algebras on the edges of $\cN$ is isomorphic to $\cM_2$ then all others are trivial. This case corresponds to a shift, which cannot have a bounded FDQC realization (see Ref.~\cite{GNVW}).\par
     Now, let us consider the case where for all $i$ one has $(\cS_{-\bm{e}_i}^{\bm{0}},\cS_{\bm{e}_i}^{\bm{0}}) \cong (\cI,\cI)$. Then $\cS_{\bm{0}}^{\bm{0}} \cong \cM_2$ otherwise the non-abelian algebra $\alpha_0(\cA_{\bm{0}})$ would contradictorily be a subalgebra of an abelian one. The conditions from cases~\ref{comm2} and~\ref{comm3} are trivially satisfied. 
     
     For the same reason, having $(\cS_{-\bm{e}_i}^{\bm{0}},\cS_{\bm{e}_i}^{\bm{0}}) \cong (\cI,\cD(\bm{n}))$ for $i$ given, implies $\cS_{\bm{0}}^{\bm{0}} \cong \cM_2$. However, in this case, the conditions of case~\ref{comm3} would fail to hold. Configuration $(\cI,\cD(\bm{n}))$ is thus not allowed.\par
     Finally, if $(\cS_{-\bm{e}_i}^{\bm{0}},\cS_{\bm{e}_i}^{\bm{0}}) \cong (\cD(\bm{n}_i), \cD(\bm{n}_i)) $ for a given $i$, from conditions~\ref{comm2} it follows that for $j\neq i$ either $(\cS_{-\bm{e}_j}^{\bm{0}},\cS_{\bm{e}_j}^{\bm{0}}) \cong (\cI, \cI)$ or $(\cD(\bm{n}_j), \cD(\bm{n}_j))$, as shown in the following. We remark that if a pair of support algebras $(\cS_{-\bm{e}_i}^{\bm{0}},\cS_{\bm{e}_i}^{\bm{0}})$ is isomorphic to trivial algebras, the QCA can be reduced to one with lower dimension $\alpha \cong \alpha'\otimes \text{id}$ with $\alpha'$ QCA on $\cA_{\Z^{s-1}}$. Therefore the relevant analysis is restricted to the case $(\cS_{-\bm{e}_i}^{\bm{0}},\cS_{\bm{e}_i}^{\bm{0}}) \cong (\cD(\bm{n}_i), \cD(\bm{n}_i)) \; \forall\,i$.  The only possibility allowed is $\bm{n}_j \parallel \bm{n}_i$, i.e. $\cD(\bm{n}_j) \cong \cD(\bm{n}_i)$, $ \forall\,i,j$. Indeed, let us consider a fixed pair $i<j$. 
    Let $\bm{\varepsilon}_k \coloneqq \pm \bm{e}_k$ with $k=i,j$. Given that, by item~\textit{1.} of Lemma~\ref{lem:overlaplemma}, we have the inclusion
\begin{equation}\label{inclusion}
    \begin{gathered}
        \cS_{\{\bm{\varepsilon}_i, \bm{\varepsilon}_j \}}^{\bm{0}} \se \cS_{\bm{\varepsilon}_i}^{\bm{0}}\otimes \cS_{\bm{\varepsilon}_j}^{\bm{0}}\cong \cD(\bm{n}_i) \otimes \cD(\bm{n}_j),\\
    \end{gathered}
\end{equation}
with 
$\{G_0, G_1\}\coloneqq\{ \bm{n}_i\cdot \bm{\sigma} \otimes I\,, I \otimes \bm{n}_j \cdot \bm{\sigma}\}$ generating the algebra $\cD(\bm{n}_i) \otimes \cD(\bm{n}_j) $,
it is convenient to impose condition \ref{comm2} case by case: when only one generator, both or neither belongs to $\cS_{\{\bm{\varepsilon}_i, \bm{\varepsilon}_j \}}^{\bm{0}}$ at fixed $\bm{\varepsilon}_i, \bm{\varepsilon}_j$. The first instance leads to a contradiction: if $G_0 \in \cS_{\{\bm{\varepsilon}_i, \bm{\varepsilon}_j \}}^{\bm{0}}$ and $G_{1} \notin \cS_{\{\bm{\varepsilon}_i, \bm{\varepsilon}_j \}}^{\bm{0}}$, then $\cS_{\bm{\varepsilon}_j }^{\bm{0}} \cong \cI$, and analogously for $G_0 \notin \cS_{\{\bm{\varepsilon}_i, \bm{\varepsilon}_j \}}^{\bm{0}}$ and $G_{1} \in \cS_{\{\bm{\varepsilon}_i, \bm{\varepsilon}_j \}}^{\bm{0}}$. If both generators belong to the support algebra (Rel.~\ref{inclusion} turns into an identity) then their commutator with any operator $T \in \cS_{\{\bm{\varepsilon}_i, \bm{\varepsilon}_j \}}^{\bm{\varepsilon}_i + \bm{\varepsilon}_j}\cong\cS_{-\{\bm{\varepsilon}_i, \bm{\varepsilon}_j\}}^{\bm0}$ must vanish. Expressing $T$ as
\begin{equation*}
    T= \sum_{a,b=0}^1 t_{ab} (\bm{n}_j \cdot \bm{\sigma})^a \otimes (\bm{n}_i \cdot \bm{\sigma})^b \quad \text{with}\; (\bm{r} \cdot \bm{\sigma})^0\coloneqq I,
\end{equation*}
we then explicitly require $[G_k,T]=0$, i.e.
\begin{equation*}
    \begin{aligned}
    0=& (1-k)\sum_{b=0}^1 t_{1b} (\bm{n}_i\times\bm{n}_j)\cdot \bm{\sigma} \otimes (\bm{n}_i \cdot \bm{\sigma})^b -\\
    &-k\sum_{a=0}^1 t_{a1} (\bm{n}_j \cdot \bm{\sigma})^a \otimes (\bm{n}_i\times\bm{n}_j)\cdot \bm{\sigma}
    \end{aligned}
\end{equation*}
for both $k=0,1$. The terms on the r.h.s.~are linearly independent and therefore they must vanish separately. The equation is satisfied for $\bm{n}_j \parallel \bm{n}_i$. Otherwise, the condition leads to $t_{01}=t_{10}=t_{11}=0$ which however brings to a contradiction: $\cS_{-\bm{\varepsilon}_i }^{\bm{0}}\cong \cS_{- \bm{\varepsilon}_j}^{\bm{0}}\cong \cI$. Finally, if none of the generators belongs to $ \cS_{\{\bm{\varepsilon}_i, \bm{\varepsilon}_j \}}^{\bm{0}}$, we have $\cS_{\{\bm{\varepsilon}_i, \bm{\varepsilon}_j \}}^{\bm{0}} \cong \left\{  a \, I \otimes I + b \, \bm{n}_i \cdot \bm{\sigma} \otimes \bm{n}_j \cdot \bm{\sigma} \; \mid \;  a,b \in \C \right\} \eqqcolon \mathcal{W}$. For any $O \in \cS_{\{\bm{\varepsilon}_i, \bm{\varepsilon}_j \}}^{\bm{0}} $ and $T$ as above we must have $\left[O\,, T \right]=0$, i.e.
\begin{equation*}
    \begin{aligned}
    0 &=t_{01}\, \bm{n}_i \cdot \bm{\sigma}\otimes (\bm{n}_j\times\bm{n}_i)\cdot \bm{\sigma} +\\
    &+  t_{10}\, (\bm{n}_i\times\bm{n}_j)\cdot \bm{\sigma} \otimes \bm{n}_j \cdot \bm{\sigma} +\\
    & + t_{11}\, (\bm{n}_i \cdot \bm{n}_j) \left[ (\bm{n}_i\times\bm{n}_j)\cdot \bm{\sigma} \otimes I - I \otimes (\bm{n}_i\times\bm{n}_j)\cdot \bm{\sigma} \right]\,.
    \end{aligned}
\end{equation*}
All the terms on the r.h.s.~must vanish separately, being linearly independent.
Once again $\bm{n}_j \parallel \bm{n}_i$ verifies the equation. If instead $\bm{n}_j \nparallel \bm{n}_i$ the two possible solutions are (i) $t_{01}=t_{10}=t_{11}=0$ and (ii) $t_{01}=t_{10}=0$ and $\bm{n}_j \perp \bm{n}_i $. Solution (i) implies $\cS_{-\{\bm{\varepsilon}_i, \bm{\varepsilon_j }\}}^{\bm{0}} \cong \cI \otimes \cI$ while solution (ii) implies $ \cS_{-\{\bm{\varepsilon}_i, \bm{\varepsilon}_j \}}^{\bm{0}} \cong \mathcal{W}$ with $\bm{n}_j \perp \bm{n}_i$. Referring to the latter case we will write $\mathcal{W}_{\perp}$ instead of $\mathcal W$ and $\bm{n}$, $\bm{n}_{\perp}$ instead of $\bm{n}_i$ and $\bm{n}_j$. Either way, all the algebras $\cS_{\{\mp\bm{\varepsilon}_i, \pm\bm{\varepsilon}_j \}}^{\bm{0}}$ must not contain any generator $G_k$ since otherwise we would have $\bm{n}_j \parallel \bm{n}_i$; therefore they can only be isomorphic to either $\mathcal{W}_{\perp}$ or $\cI \otimes \cI$. Since the following holds
\begin{subequations}
  \begin{empheq}[left=\empheqlbrace]{align}
    &  \, \alpha_0(\cA_{\bm{0}}) \se \cS_{\bm{0}}^{\bm{0}} \otimes \bigotimes_{j>i}\left( \cS_{\{\bm{\varepsilon}_i, \bm{\varepsilon}_j \}}^{\bm{0}} \otimes \cS_{-\{\bm{\varepsilon}_i, \bm{\varepsilon}_j \}}^{\bm{0}} \label{sys a}\right)\\
    &  \,  \alpha_0(\cA_{\bm{0}}) \se \cS_{\bm{0}}^{\bm{0}} \otimes \bigotimes_{j>i}\left( \cS_{\{-\bm{\varepsilon}_i, \bm{\varepsilon}_j \}}^{\bm{0}} \otimes \cS_{\{\bm{\varepsilon}_i, -\bm{\varepsilon}_j \}}^{\bm{0}} \label{sys b}\right),
  \end{empheq}
\end{subequations}
the case where any of the factors on the r.h.s.~of Rel.s~\ref{sys a} and \ref{sys b} is $\cI\otimes\cI$ is already excluded, as
we showed that for a fixed $\bm e_i$ the configurations $(\cI,\cD(\bm n))$ and $(\cD(\bm n),\cI)$ are excluded.
Alternatively, i.e. if $ \cS_{\{\bm{\varepsilon}_i, \bm{\varepsilon}_j \}}^{\bm{0}}\cong  \mathcal{W}_{\perp}\; \forall \, \bm{\varepsilon}_i, \bm{\varepsilon}_j $, the image of an operator $O \in \cA_{\bm{0}}$ under the local update rule is given by
\begin{equation*}
            \alpha_0(O)= \begin{matrix}
                    & & I & &\\
                    & & \otimes & &\\
                    I & \otimes & A(O) & \otimes & I\\
                    & & \otimes & &\\
                    & & I & &
                \end{matrix}
                +  \begin{matrix}
                    & & \bm{n}_{\perp} \cdot \bm{\sigma} & &\\
                    & & \otimes & &\\
                     & \bm{n} \cdot \bm{\sigma} & \otimes\;\;B(O) \;\;\otimes&   \bm{n} \cdot \bm{\sigma}&\\
                    & & \otimes & &\\
                    & & \bm{n}_{\perp} \cdot \bm{\sigma} & &
                \end{matrix}\,,
        \end{equation*}
where  $A,B$ are linear maps from $\cM_2$ to algebra of $\cN\setminus\{\bm e_i,\bm e_j,-\bm e_i,-\bm e_j\}$.
We observe that the linear map $B(\cdot)$ cannot be null, otherwise $\alpha_0$ would be again a site-wise unitary leading to the unacceptable case $\cS_{\pm\bm{e}_i}^{\bm{0}} \cong \cS_{\pm\bm{e}_j}^{\bm{0}} \cong \cI$. Let us now consider the commutators
\begin{align}
   & \left[ \alpha(\cA_{\bm{0}}), \alpha(\cA_{\bm{e}_i}) \right]=0 \label{hor}\,, \\
   & \left[ \alpha(\cA_{\bm{0}}), \alpha(\cA_{\bm{e}_j}) \right]=0 \label{vert}\,.
\end{align}
The algebras $\alpha(\cA_{\bm{0}})$ and $\alpha(\cA_{\bm{e}_i})$ overlap on sites $\{\bm{0},\bm{e}_i\}$, whereas the algebras $\alpha(\cA_{\bm{0}})$ and $\alpha(\cA_{\bm{e}_j})$ overlap on $\{\bm{0},\bm{e}_j\}$. We can focus just on these sites since the commutation of the algebras is trivial elsewhere. Let us start with the first (\ref{hor}). For any $L,R \in \cA_{\bm{0}}$ we have
\begin{equation*}
    \begin{aligned}
            & 0=\left[ \alpha_0(L), \tau^{\bm{e}_i}\alpha_0(R) \right] =\\
            &  \begin{pmatrix}
                    & & I & \bm{n}_{\perp} \cdot \bm{\sigma} & \\
                    & & \otimes & \otimes &\\
                     & I & \otimes\;\;[A(L)\,,\,\bm{n} \cdot \bm{\sigma} ] \;\;\otimes&  B(R) & \otimes \; \bm{n} \cdot \bm{\sigma}\\
                    & & \otimes & \otimes &\\
                    & & I & \bm{n}_{\perp} \cdot \bm{\sigma} &
                \end{pmatrix} +\\
                &  \begin{pmatrix}
                    & & \bm{n}_{\perp} \cdot \bm{\sigma} & I & \\
                    & & \otimes & \otimes &\\
                     &  \bm{n} \cdot \bm{\sigma} & \otimes\;\; B(L)\;\;\otimes&  [\bm{n} \cdot \bm{\sigma}\,,\,A(R) ] & \otimes \; I\\
                    & & \otimes & \otimes &\\
                    & & \bm{n}_{\perp} \cdot \bm{\sigma} & I &
                \end{pmatrix} +\\
                &  \frac{1}{2} \left(  \begin{matrix}
                    & & \bm{n}_{\perp} \cdot \bm{\sigma} & \bm{n}_{\perp} \cdot \bm{\sigma} & \\
                    & & \otimes & \otimes &\\
                     & \bm{n} \cdot \bm{\sigma} & \otimes\;\;[B(L)\,,\,\bm{n} \cdot \bm{\sigma} ] \;\;\otimes&   \{\bm{n} \cdot \bm{\sigma}\,,\,B(R) \} & \otimes \; \bm{n} \cdot \bm{\sigma}\\
                    & & \otimes & \otimes &\\
                    & & \bm{n}_{\perp} \cdot \bm{\sigma} & \bm{n}_{\perp} \cdot \bm{\sigma} &
                \end{matrix} + \right.\\
                & \left. \begin{matrix}
                    & & \bm{n}_{\perp} \cdot \bm{\sigma} & \bm{n}_{\perp} \cdot \bm{\sigma} & \\
                    & & \otimes & \otimes &\\
                     & \bm{n} \cdot \bm{\sigma} & \otimes\;\; \{B(L)\,,\,\bm{n} \cdot \bm{\sigma} \} \;\;\otimes&   [\bm{n} \cdot \bm{\sigma}\,,\,B(R) ] & \otimes \; \bm{n} \cdot \bm{\sigma}\\
                    & & \otimes & \otimes &\\
                    & & \bm{n}_{\perp} \cdot \bm{\sigma} & \bm{n}_{\perp} \cdot \bm{\sigma} &
                \end{matrix} \right) \,,
        \end{aligned}
\end{equation*}
where we used the identity
        \begin{equation*}\label{comm simm}
            [X \otimes Y, Z \otimes W]=\frac{1}{2}\left( [X,Z]\otimes \{Y,W\} +  \{X,Z\} \otimes [Y, W]\right)\,.
        \end{equation*}
The three terms are linearly independent as the operators at the edges differ, so they must cancel separately. Hence for any operator $O,L,R \in \cM_2  $
\begin{subnumcases}{}
  &  $ [A(O)\,,\,\bm{n} \cdot \bm{\sigma} ]=0 $\label{subnum ax}
   \\[2mm]
   & $ \begin{aligned}
       &  [B(L)\,,\,\bm{n} \cdot \bm{\sigma} ] \otimes  \{\bm{n} \cdot \bm{\sigma}\,,\,B(R) \}\, +  \\
    &  + \{B(L)\,,\,\bm{n} \cdot \bm{\sigma} \} \otimes [\bm{n} \cdot \bm{\sigma}\,,\,B(R) ]=0\,.
   \end{aligned} $ \label{subnum bx}
\end{subnumcases}
The two terms of Eq.~\ref{subnum bx} are orthogonal in the Frobenius product and thus must cancel independently. If $[B(L)\,,\,\bm{n} \cdot \bm{\sigma} ]=0$, then $[\bm{n} \cdot \bm{\sigma}\,,\,B(R) ]=0$ since $ \{B(L)\,,\,\bm{n} \cdot \bm{\sigma} \}$ cannot be null. For the same reason, if $\{B(L)\,,\,\bm{n} \cdot \bm{\sigma} \}=0$, then $\{\bm{n} \cdot \bm{\sigma}\,,\,B(R) \}=0$. We refer to the former eventuality by $\mathscr{C}^{\bm{n}}$ and to the latter by $\mathscr{A}^{\bm{n}}$.\\
Analogously, for any operator $O,L,R \in \cM_2  $, from commutation \ref{vert} we conclude
\begin{subnumcases}{}
  &  $ [A(O)\,,\,\bm{n}_{\perp}  \cdot \bm{\sigma} ]=0 $\label{subnum az}
   \\[2mm]
   & $ \begin{aligned}
       &  [B(L)\,,\,\bm{n}_{\perp}  \cdot \bm{\sigma} ] \otimes  \{\bm{n}_{\perp}  \cdot \bm{\sigma}\,,\,B(R) \}\, +  \\
    &  + \{B(L)\,,\,\bm{n}_{\perp}  \cdot \bm{\sigma} \} \otimes [\bm{n}_{\perp}  \cdot \bm{\sigma}\,,\,B(R) ]=0\,,
   \end{aligned} $ \label{subnum bz}
\end{subnumcases}
where Eq.~\ref{subnum bz} yields analogs conditions $\mathscr{C}^{\bm{n}_{\perp}}$, $\mathscr{A}^{\bm{n}_{\perp}}$ with $\bm{n}_{\perp}$ instead of $\bm{n}$. For any $ O \in \cM_2$, if we assume either $\mathscr{C}^{\bm{n}} \land \mathscr{C}^{\bm{n}_{\perp}}$ or $\mathscr{A}^{\bm{n}} \land \mathscr{A}^{\bm{n}_{\perp}}$, it follows that $B(O) \propto I$; if we assume $\mathscr{C}^{\bm{n}} \land \mathscr{A}^{\bm{n}_{\perp}}$ we obtain $B(O) \propto  \bm{n} \cdot \bm{\sigma}$; lastly, if we assume $\mathscr{A}^{\bm{n}} \land \mathscr{C}^{\bm{n}_{\perp}}$ we get $B(O)\propto \bm{n}_{\perp}  \cdot \bm{\sigma}$.
From \ref{subnum ax} and \ref{subnum az} it follows that $A(O) \propto I$. 
In any case, $\alpha_0(\cA_{\bm{0}})$ turns out to be, contradictorily, an abelian algebra.
 \end{proof}
\subsection{Proof of Theorem~1}
\begin{proof}
 In case \textit{i.} of Lemma \ref{lem:confsuppalg} the local update rule $\alpha_0$ must realize either a shift of the algebra $\cA_{\bm{0}}$ by one site along $\pm \bm{e}_i$ (if $\bm x=\pm{\bm e}_i$), possibly combined with a site-wise unitary, or just a site-wise unitary (if $\bm x=\bm 0$).
 In case \textit{ii.} of Lemma \ref{lem:confsuppalg}, after a site-wise unitary that performs a basis change, the image of any operator $O \in \cA_{\bm{0}}$ will generally be of the form
\begin{equation*}\label{eq: image gen}
\alpha_0(O)=\sum_{\bm a, \bm b}\dyad{\bm a}_+\otimes                    
                    U^\dag_{\bm a\bm b}\, O\, U_{\bm a\bm b}\otimes\dyad {\bm b}_{-},
\end{equation*}
where $\dyad{\bm a}_+=\bigotimes_{j=1}^s\dyad{a_j}_{{\bm e}_j}$ and 
$\dyad{\bm b}_-=\bigotimes_{j=1}^s\dyad{b_j}_{-{\bm e}_j}$. 
For fixed values of $a_i,b_i$ in all coordinate directions except one, the local update rule must 
obey the same constraints as that of a QCA on $\mathbb Z$, which, according to 
Ref.~\cite{schumacher2004reversible}, is given by a local basis change followed by a multiply
controlled-phase in the directions $\bm{e}_i,-\bm{e}_i$. It must then be
\begin{equation*}
    \begin{gathered}
              U_{\bm a\bm b}=\exp{\frac{i}{2}\left[(a_i+b_i)\phi_{\bar{\bm a}_i\bar{\bm b}_i}\right](I-\sigma^3)}W_{\bar{\bm a}_i\bar{\bm b}_i}\,,
    \end{gathered}
\end{equation*}
where $\bar{\bm v}_i$ denotes the $s-1$-tuple obtained from the $s$-tuple $\bm v$ by 
removing the $i$-th component. The same description must hold if we consider a different 
coordinate direction, say $\bm e_j$. Then
\begin{equation*}
    \begin{gathered}
              U_{\bm a\bm b}=\exp{\frac{i}{2}\left[(a_j+b_j)\phi_{\bar{\bm a}_j\bar{\bm b}_j}\right](I-\sigma^3)}W_{\bar{\bm a}_j\bar{\bm b}_j}\,.
    \end{gathered}
\end{equation*}
Clearly, $W_{\bar{\bm a}_j\bar{\bm b}_j}^\dag W_{\bar{\bm a}_i\bar{\bm b}_i}$ must be diagonal in the eigenbasis of $\sigma^3$, and thus we can set 
$W_{\bar{\bm a}_i\bar{\bm b}_i}=\exp{\frac{i}{2}\lambda_{\bar{\bm a}_i\bar{\bm b}_i}(I-\sigma^3)}W_0$. From now on, we will disregard $W_0$ that can be included in the initial local change of basis. Consequently, we can write
\begin{equation*}
    \begin{gathered}
              U_{\bm a\bm b}=\exp{\frac{i}{2}\left[(a_i+b_i)\phi_{\bar{\bm a}_i\bar{\bm b}_i}+ \lambda_{\bar{\bm a}_i\bar{\bm b}_i}\right](I-\sigma^3)}\,.
    \end{gathered}
\end{equation*}  
The two descriptions obtained considering different directions $\bm e_i$ and $\bm e_j$ are 
equivalent and therefore the rotation angles must coincide modulo $2k\pi$, for $k \in \Z$
\begin{equation}\label{sys eqs angles}
    (a_i+b_i)\phi_{\bar{\bm a}_i\bar{\bm b}_i} + \lambda_{\bar{\bm a}_i\bar{\bm b}_i}= (a_j+b_j)\phi_{\bar{\bm a}_j\bar{\bm b}_j} + \lambda_{\bar{\bm a}_j\bar{\bm b}_j}\,.
\end{equation}
The system of equations (\ref{sys eqs angles}) is solved by
\begin{align*}
    &\lambda_{\bar{\bm a}_i\bar{\bm b}_i}=(a_j+b_j)\psi_{\bar{\bm a}_{ij}\bar{\bm b}_{ij}}+\mu_{\bar{\bm a}_{ij}\bar{\bm b}_{ij}}\,,\\
     &\phi_{\bar{\bm a}_i\bar{\bm b}_i}=(a_j+b_j)\chi_{\bar{\bm a}_{ij}\bar{\bm b}_{ij}}+\nu_{\bar{\bm a}_{ij}\bar{\bm b}_{ij}}\,,
\end{align*} 
where $\bar{\bm v}_{ij}$ represents the $s-2$-tuple obtained from $\bm v$ by removing the 
$i$-th and $j$-th component, $\psi_{\bar{\bm a}_{ij}\bar{\bm b}_{ij}}$ is equal to 
$\phi_{\bar{\bm a}_j\bar{\bm b}_j}$ with $a_i=b_i=0$, and 
$\mu_{\bar{\bm a}_{ij}\bar{\bm b}_{ij}}$ is equal to 
$\lambda_{\bar{\bm a}_j\bar{\bm b}_j}$ with the same condition, while similarly 
$\chi_{\bar{\bm a}_{ij}\bar{\bm b}_{ij}}$ is equal to 
$\phi_{\bar{\bm a}_j\bar{\bm b}_j}$ with $a_i=0$, $b_i=1$ or $a_i=1$, $b_i=0$, and 
$\nu_{\bar{\bm a}_{ij}\bar{\bm b}_{ij}}$ is equal to 
$\lambda_{\bar{\bm a}_j\bar{\bm b}_j}$ with the same conditions. 
By imposing that the following commutators vanish (as per 
Lemma~\ref{lem:iff_werner}), with fixed $\bar{\bm b}_i,\bar{\bm a}'_i$ and  
$\bar{\bm b}_j,\bar{\bm a}'_j$, respectively
\begin{gather*}
    \sum_{b_i,a'_i}\left[ U_{\bm a\bm b}^\dag\, \sigma_\pm U_{\bm a\bm b} \otimes \dyad{b_i}\,,\, \dyad{a'_i} \otimes U_{\bm a'\bm b'}^\dag\, \sigma_\pm U_{\bm a'\bm b'} \right]\,,\\
        \sum_{b_j,a'_j}\left[ U_{\bm a\bm b}^\dag\, \sigma_\pm U_{\bm a\bm b} \otimes \dyad{b_j}\,,\, \dyad{a'_j} \otimes U_{\bm a'\bm b'}^\dag\, \sigma_\pm U_{\bm a'\bm b'} \right]\,,
\end{gather*}
where $\sigma_{\pm} \coloneqq (\sigma_1 \pm i \sigma_2)$, we obtain 
\begin{equation*}
    \phi_{\bar{\bm a}_i\bar{\bm b}_i}=\phi_{\bar{\bm a}'_i\bar{\bm b}'_i}\quad \text{and}\quad \phi_{\bar{\bm a}_j\bar{\bm b}_j}=\phi_{\bar{\bm a}'_j\bar{\bm b}'_j}\,,
\end{equation*} 
i.e., these angles are independent of the operators chosen on the edges.
Hence we conclude 
\begin{equation*}
    U_{\bm a\bm b}= \exp{\frac{i}{2}\left[\sum_{i=1}^s(a_i+b_i)\,\varphi_i\right]\,(I-\sigma^3)} 
\end{equation*}
with $\varphi_i \in [0,2\pi]$, that is the composition of $2s$ two-qubit controlled-phases of angle $\varphi_i$ along directions $\pm\bm{e}_i$.
\end{proof}

\section{Section C: Index Classification}
\red{The two-layer Finite-Depth Quantum Circuits (FDQCs) discussed in the body of the manuscript are known in the literature as Margolus FDQCs \cite{schumacher2004reversible}}.
\begin{definition}\label{def:margolus}
    \red{A \emph{Margolus FDQC (MFDQC)} is a depth-$2$ quantum circuit $\alpha=\upsilon_2 \cdot \upsilon_1$ whose gates act on the Margolus partition of the lattice (FIG.~\ref{fig:Margolus}), i.e.~$\{\Lambda_{\bm j}^1\}=\{\mathfrak{c}~+~2\bm{j}\}_{\bm j \in \Z^s} $and $\{\Lambda_{\bm j}^2\}=\{\mathfrak q(\bm q)+2\bm j\}_{\bm j \in \Z^s}$ for any one of the $\mathfrak q(\bm q) \in Q$.}
\end{definition}
\noindent
\red{We notice that the realization as an MFDQC resembles the one of a Partitioned QCA (PQCA)~\cite{arrighi2012intrinsically, arrighi2019overview} but with a crucial difference: in a PQCA all the gates of the implementation must be the same.}
\begin{figure}[h]
    \centering
    \includegraphics[width=6.2cm]{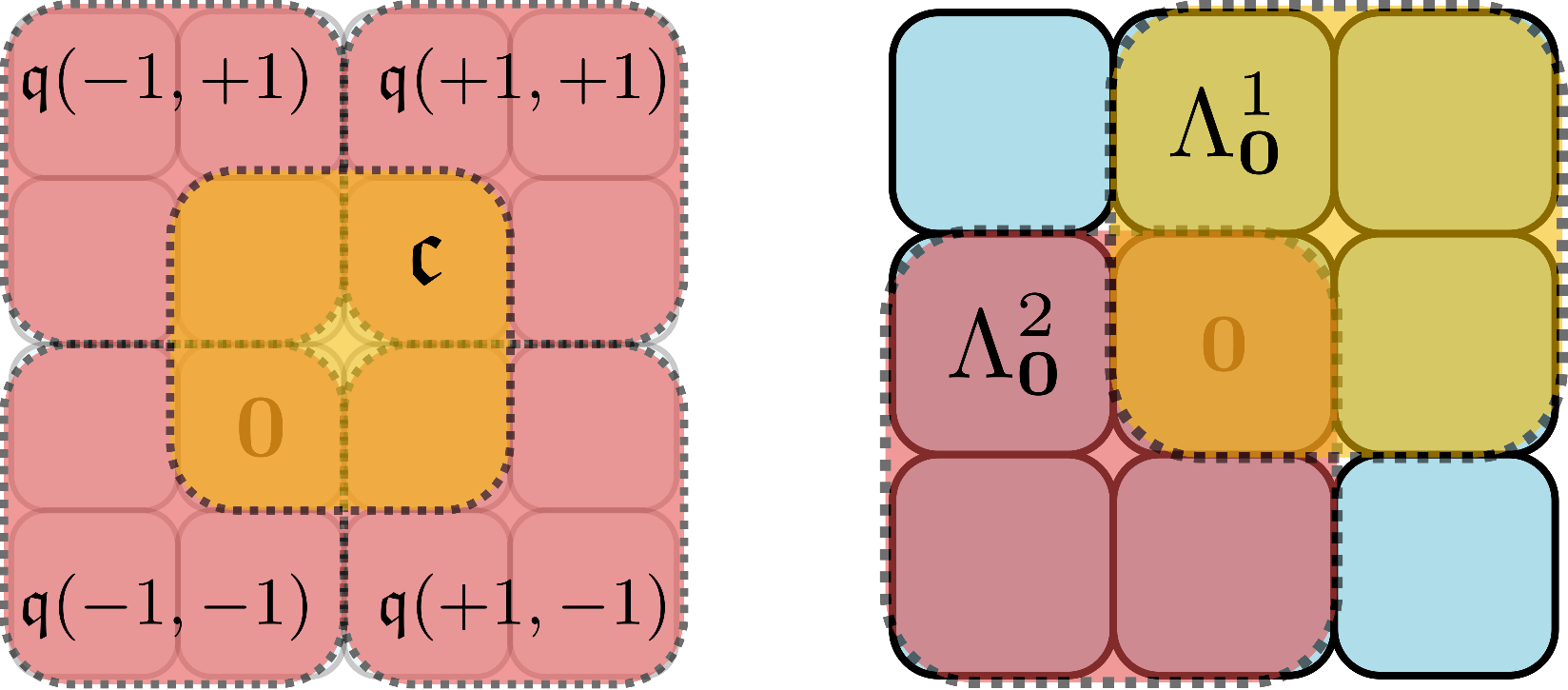}
    \caption{On the left are depicted the super-cell $\mathfrak{c}$ and quadrant sets $\mathfrak q(\bm q)$ on $\Z^2$. On the right are displayed the primitive cells for the Margolus partitioning scheme of $\Z^2$. A Margolus FDQC (MFDQC) consists of the yellow layer followed by the red one.}
    \label{fig:Margolus}
\end{figure}

The minimal, albeit partial, extension to higher spatial dimensions of the one-dimensional index theory is given by the following Lemma~\ref{lem:Vogts}. Its content is a straightforward generalization of known results from Ref.s~\cite{schumacher2004reversible, Vogts2009DiscreteTQ}. The related proof descends directly from them and will therefore be omitted in the following discussion.
\begin{lemma}\label{lem:Vogts}
    Let $\alpha$ be a QCA of a lattice $\Z^s$ of qudits $\cA_x \cong \cM_{d}$ with a neighborhood scheme enclosed within the Moore scheme $\cN \se \{\bm{0}\} \cup  \{ \pm \bm{e_i} \}_{i=1}^s \cup \{ \bm{e_i} \pm \bm{e_j} \}_{i,j=1}^s \cup \{ -\bm{e_i} \pm \bm{e_j} \}_{i,j=1}^s$. If the support algebras of $\alpha(\cA_\mathfrak{c})$ on quadrants $\mathfrak q(\bm q)$ commute as subalgebras of $\cA_\mathfrak{c}$
    \begin{equation}\label{Vogts hyp_sm}
       \left[ \tau^{-\bm{q}}\,\cS_{\mathfrak q(\bm q)}^{\mathfrak{c}},\, \tau^{-\bm{p}}\, \cS_{\mathfrak q({\bm{p}})}^{\mathfrak{c}} \right]=0\quad \forall\,\mathfrak q(\bm q),\mathfrak q({\bm{p}}) \in Q\,, 
    \end{equation}
    then 
    \begin{equation*}
        \alpha(\cA_{\mathfrak{c}} ) = \bigotimes_{\mathfrak q(\bm q) \in Q} \cS^{\mathfrak{c}}_{\mathfrak q(\bm q)} \cong \bigotimes_{\mathfrak q(\bm q)\in Q} \cM_{d(\mathfrak q(\bm q))}\,,
    \end{equation*}
     with $\prod_{\mathfrak q(\bm q)\in Q}d(\mathfrak q(\bm q)) =d(\mathfrak{c})= d^{2^s} $. The QCA $\alpha$ is an MFDQC if and only if $\forall\, \mathfrak q(\bm q) \in Q\;\; d(\mathfrak q(\bm q))=d $.
\end{lemma}
To establish a stronger result we focused on a particular subclass of these QCA.
\begin{definition}
Let $\mathscr{V}_p^s$ be the sets of von Neumann neighborhood QCA of a hypercubic lattice $\Z^s$ with $\cA_{\bm{x}} \cong \cM_p$, $p$ prime, for which Eq.~\eqref{Vogts hyp_sm} is satisfied.
\end{definition}
\noindent
This class of QCA is the natural extension of those of qubits.
Indeed, a sufficient condition \cite{Vogts2009DiscreteTQ} for commutation relations (\ref{Vogts hyp_sm}) to hold is
\begin{equation}\label{Vogts hyp strong}
       \left[ \tau^{-\bm{e_i}}\,\cS_{\bm{e_i}}^{\bm{0}},\, \tau^{-\bm{e_j}}\, \cS_{\bm{e_j}}^{\bm{0}} \right]=0\;\; \forall\,i,j\,. 
    \end{equation}
According to the classification provided by Lemma \ref{lem:confsuppalg}, our qubit QCA satisfy condition (\ref{Vogts hyp strong}) and thus are of type $\mathscr{V}^s_2$.\\ 
For the sets $\mathscr{V}_p^s$, it is profitable to define analogs of the one-dimensional index along Cartesian directions of the lattice. The components $i=1,\ldots,s$ of the \emph{index}
$\bm{\iota}(\alpha)$ of a QCA $\alpha \in \mathscr{V}_p^s $ are given by
\begin{equation}\label{index comp_sm}
        \iota_{i}(\alpha) \coloneqq \sqrt[\leftroot{-3}\uproot{3}2^{s-1}]{ \frac{ \prod_{\mathfrak q(\bm q) \in R_{i}} d\left(\mathfrak q(\bm q)\right)} { p^{2^{s-1}} } } \, \in \mathbb{Q}_+\,,
\end{equation}
where 
$R_{i} \coloneqq \left\{ \mathfrak q(\bm q) \in Q\mid q_i=-(\bm{e_i})_i \right\} $. In Heisenberg's picture, we recall that an evolution representing, e.g.,  a shift to the right by $\bm x$, shifts operators to the left by $-\bm x$. Hence, an index accounting for information flow to the right through some boundary addresses algebras to the left of that boundary (FIG.~\ref{fig:dimsupp}(A)).

\begin{lemma}\label{lem:lemmas}
    Let $\alpha \in \mathscr{V}_p^s\,$, so that
    \begin{equation}\label{isom nd}
        \alpha(\cA_{\mathfrak{c}} ) = \bigotimes_{\mathfrak q(\bm q) \in Q} \cS^{\mathfrak{c}}_{\mathfrak q(\bm q)} \cong \bigotimes_{\mathfrak q(\bm q)\in Q} \cM_{d(\mathfrak q(\bm q))}\,,
    \end{equation}
    with $ \prod_{\mathfrak q(\bm q)\in Q}d(\mathfrak q(\bm q))=p^{2^s} $ and $d(\mathfrak q(\bm q))=p^{n\left(\mathfrak q(\bm q)\right)}$. Then
     \begin{enumerate}
        \item $n\left(\mathfrak q(\bm q)\right) \in \{0, 1, 2\}\,$,
        \item $n\left(\mathfrak q(\bm q)\right)$ have a fixed parity $\forall\,\mathfrak q(\bm q) \in Q\,$,
\item[3.] $\bm{\iota}(\alpha)=\bm{1}$ iff $n\left(\mathfrak q(\bm q)\right)=1 \; \forall\, \mathfrak q(\bm q) \in Q\,$.
    \end{enumerate}
\end{lemma}
\noindent
See Figure~\ref{fig:dimsupp}(A).
\begin{proof}
For the sake of simplicity, we present the proof for $s=2$. By the von Neumann neighborhood scheme, the images of the algebras of any two sites can overlap only on a given plane of the lattice. The proof for $s>2$ requires going through a broader casuistry, but the arguments employed remain exactly the same.\\
    (\textit{1.}) 
    From Eq.~(\ref{isom nd}) follows $\sum_{\mathfrak q(\bm q) \in Q} n(\mathfrak q(\bm q))=4$ and $n(\mathfrak q(\bm q)) \in \{0,\ldots,4\}$. We show that $n(\mathfrak q(\bm q))$ cannot be equal to $4$ nor to $3$. For instance, let us consider $d\left(\mathfrak{q} (+1,+1)\right)=p^{n\left(\mathfrak q (+1,+1)\right)}$.  To have $n\left(\mathfrak q (+1,+1)\right)=4$  all the algebras $\cA_{\bm{x}} \cong \cM_{p}$ associated to the four sites belonging to the super-cell $\mathfrak{c} \ni \bm{x}$ should end up contributing to $\cS_{\mathfrak q (+1,+1)}^{\mathfrak{c}}$. 
    This is impossible since the neighborhood scheme prevents the algebra $\cA_{\bm{0}}$ from being mapped to any subalgebra of $\cA_{\mathfrak q (+1,+1)}$.
    Likewise, to have $n_j\left(\mathfrak q(+1,+1)\right)=3$ all the algebras $\cA_{\bm{e_1}}, \cA_{\bm{e_2}}$ and $ \cA_{\bm{e_1}+\bm{e_2}}$ should be mapped to $\cA_{ \mathfrak q(+1,+1)}$. The only way for this to happen, according to the neighborhood scheme, is for $\cA_{\bm{e_2}}$ to be shifted along $\bm{e_1}$ and for $\cA_{\bm{e_1}}$ to be shifted along $\bm{e_2}$ (modulo a site-wise unitary). However, for translation invariance, if one is shifted to one direction the same must happen to the other.\\
    (\textit{2.})  
    By item~\textit{5.} of Lemma \ref{lem:overlaplemma}, in terms of generating algebras we have (FIG.~\ref{fig:supp_gen})
    \begin{equation}\label{genn}
     \begin{aligned}
    \cS_{\mathfrak q(+1,+1)}^{\mathfrak{c}} \otimes &\cS_{\mathfrak q(+1,-1)}^{\mathfrak{c}} = \\ &\cS^{\bm{0}}_{\bm{e_1}} \lor \cS^{\bm{e_2}}_{\bm{e_1}+\bm{e_2}}   \lor  \\
    &  \cS^{\bm{e_1}}_{\{\bm{e_1},\, 2\bm{e_1},\, \bm{e_1}+\bm{e_2},\, \bm{e_1}-\bm{e_2}  \}} \lor \\
    &  \cS^{\bm{e_1}+\bm{e_2}}_{\{ \bm{e_1}+\bm{e_2},\, 2\bm{e_1}+\bm{e_2},\, \bm{e_1}+2\bm{e_2},\, \bm{e_1} \}}  \,.
\end{aligned}   
    \end{equation}

\begin{figure}[h]
    \centering
    \includegraphics[width=8.0cm]{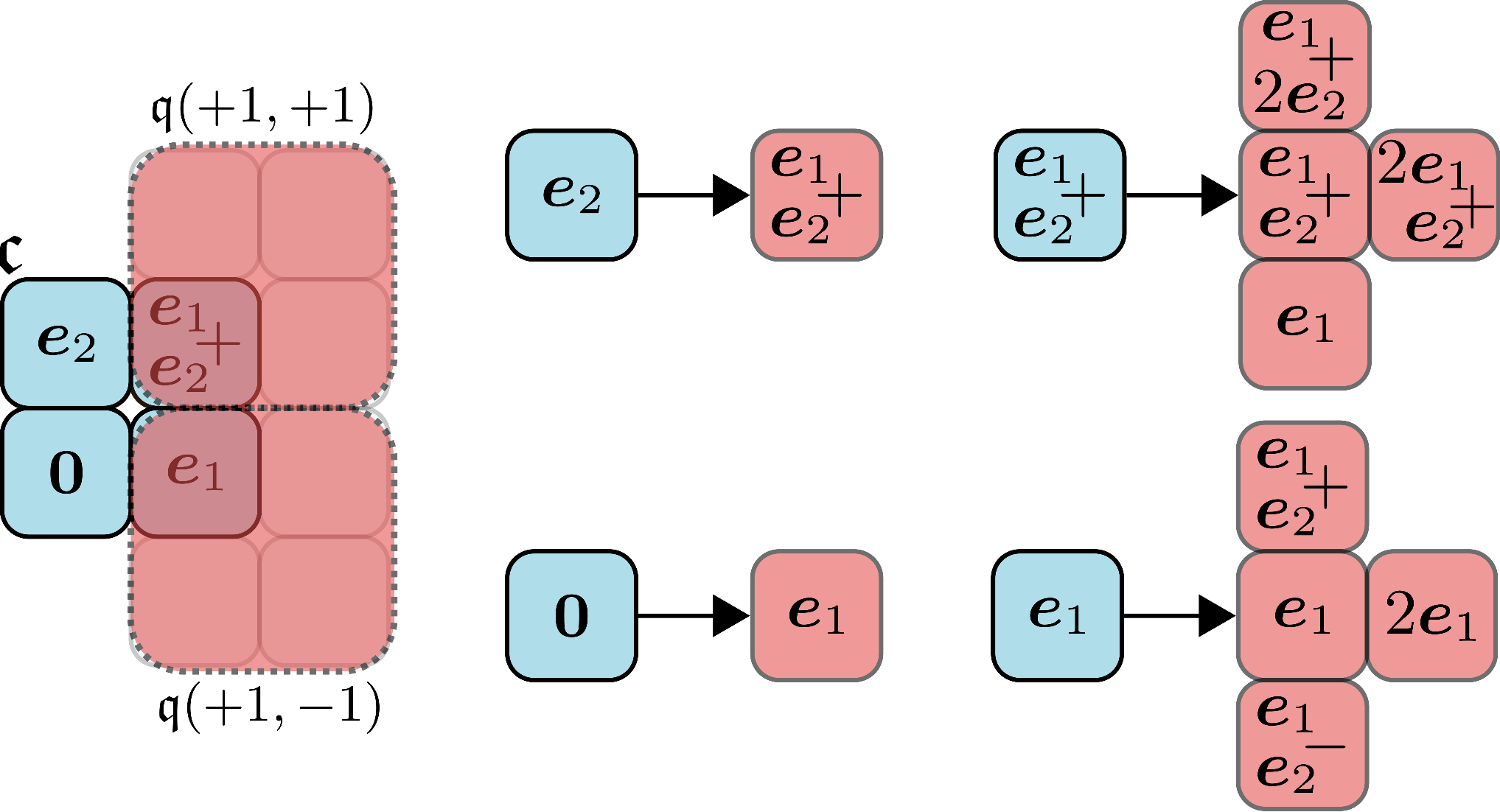}
    \caption{The Figure shows the contributions of $\alpha(\cA_{\mathfrak{c}})$ to the support algebra $ \cS_{\mathfrak q(+1,+1)}^{\mathfrak{c}} \otimes \cS_{\mathfrak q(+1,-1)}^{\mathfrak{c}}$ in terms of the related generating support algebras (\ref{genn}).}
    \label{fig:supp_gen}
\end{figure}
\noindent
By item~\textit{4.} of Lemma \ref{lem:overlaplemma} follows that the first two terms and the second two on the r.h.s.~of Eq.~(\ref{genn}) are respectively isomorphic. Thus the algebra on the l.h.s.~of Eq.~(\ref{genn}) receives the same contribution from the first two terms, as well as it receives the same contribution from the second two terms. This implies  $n\left(\mathfrak q(+1,+1)\right)+n\left(\mathfrak q(+1,-1)\right)=2m$ with $m\in \mathbb{N}_0$. Hence $n\left(\mathfrak q(+1,+1)\right)$ and $ n\left(\mathfrak q(+1,-1)\right)$ must be both even or odd integers. Analogously for $n\left(\mathfrak q(+1,+1)\right),n\left(\mathfrak q{(-1,+1)}\right)$ and for $n\left(\mathfrak q(+1,-1)\right),n\left(\mathfrak q{(-1,-1)}\right)$. It follows that  all the $n(\mathfrak q(\bm q))$ must be either even or odd integers.\\
(\textit{3.})
From previous items \textit{1.} and \textit{2.} of Lemma \ref{lem:lemmas} follows the well-posedness of Eq.~(\ref{index comp_sm}).\\
Now let $\iota_{1}(\alpha)=\iota_{2}(\alpha)=1$, then
\begin{align}
    &n\left(\mathfrak q(\bm q)\right)=1 \quad \forall\; \mathfrak q(\bm q) \in Q\,,  \label{ok}\\
    &\text{or}  \quad \begin{cases} \, n\left(\mathfrak q(q_i,q_j)\right)=n\left(\mathfrak q(-q_i,-q_j)\right)=2\\
            \, n\left(\mathfrak q(-q_i,q_j)\right)=n\left(\mathfrak q(q_i,-q_j)\right)=0 
            \end{cases}\,. \label{patologico}
\end{align}
We show that there is no QCA $\alpha \in \mathscr{V}_p^2$ such that Eq.s (\ref{patologico}) hold.
Suppose  
\begin{equation}\label{full 2}
 \begin{aligned}
    &\cS_{\mathfrak q(q_i,q_j)}^{\mathfrak{c}} \cong \cS_{\mathfrak q(-q_i,-q_j)}^{\mathfrak{c}} \cong \cM_{p^2}\,, \\
    &\cS_{\mathfrak q(-q_i,q_j)}^{\mathfrak{c}} \cong \cS_{\mathfrak q(q_i,-q_j)}^{\mathfrak{c}} \cong \cI\,.
\end{aligned}   
\end{equation}
 Without loss of generality let $q_i=q_j=+1$, then by item~\textit{5.} of Lemma~\ref{lem:overlaplemma} follows
\begin{equation}\label{algebra id}
    \cI \cong \cS_{\mathfrak q(-1,+1)}^{\mathfrak{c}}= \cS_{\bm{e_2}}^{\bm{e_1}+\bm{e_2}} \lor  \cS_{\bm{e_2}}^{\bm{0}} \lor  \cS_{\{\bm{e_2}, \, \bm{e_2}-\bm{e_1},\, 2\bm{e_2} \}}^{\bm{e_2}}\,.
\end{equation}
By definition, a QCA preserves trivial algebras. Since the support algebras on the r.h.s. of the Eq.~(\ref{algebra id}) generate a trivial algebra, they must be themselves trivial algebras. Hence by items~\textit{2.} and \textit{4.} of Lemma \ref{lem:overlaplemma} any operator $A \in \alpha_0(\cA_{\bm{0}})$ takes the form 
\begin{equation}\label{id left}
    A = I_{\{ \bm{0}, -\bm{e_1},\, \bm{e_2} \}} \otimes O_{\{ \bm{e_1},\, -\bm{e_2}  \}\,,}
\end{equation}
 with $O_{\{ \bm{e_1},\, -\bm{e_2}  \}} \in \cS_{\{ \bm{e_1},\, -\bm{e_2}  \}}^{\bm{0}}$. On the other hand, repeating the same argument for $\cS_{\mathfrak q(+1,-1)}^{\mathfrak{c}}$ yields 
 \begin{equation}\label{id right}
      A = I_{\{ \bm{0}, \bm{e_1},\, -\bm{e_2} \}} \otimes O_{\{ -\bm{e_1},\, \bm{e_2} \}\,,} 
 \end{equation}
with $O_{\{ -\bm{e_1},\, \bm{e_2}  \}} \in \cS_{\{ -\bm{e_1},\, \bm{e_2}  \}}^{\bm{0}}$. Therefore, combining Eq.s (\ref{id left}) and (\ref{id right}), $\alpha_0$ turns out to map any operator to the identity; consequently, it cannot give rise to (\ref{full 2}). Hence $\bm{\iota}(\alpha)=\bm{1}$ necessarily implies (\ref{ok}). The converse implication can be readily verified through direct computation of the index formula (\ref{index comp_sm}).
\end{proof}
\subsection{Proof of Theorem~2}
\begin{proof}
Following the construction of Ref.~\cite{GNVW}, we start by showing that if two QCA have the same index then they are equal modulo an FDQC.
     Let $\alpha_1, \alpha_2 \in \mathscr{V}_p^s$ and assume $\bm{\iota}(\alpha_1)=\bm{\iota}(\alpha_2)$. Let us refer with $~^J\!\cS$ to the support algebras of the first $(J=1)$ and second $(J=2)$ QCA respectively. Lemma~\ref{lem:lemmas} guarantees that
     \begin{gather*}
         \alpha_J(\cA_{\mathfrak{c}} ) = \bigotimes_{\mathfrak q(\bm q) \in Q} ~^J\!\cS^{\mathfrak{c}}_{\mathfrak q(\bm q)} \cong \bigotimes_{\mathfrak q(\bm q)\in Q} \cM_{d_J(\mathfrak q(\bm q))}\,,\\
         \cA_{\mathfrak{c}} = \bigotimes_{\mathfrak q(\bm q) \in Q} ~^J\!\cS_{\mathfrak{c}}^{\mathfrak q(\bm q)} \cong \bigotimes_{\mathfrak q(\bm q)\in Q} \cM_{d_J(-\mathfrak q(\bm q))}\,,
     \end{gather*} with
    $ \prod_{\mathfrak q(\bm q)\in Q}d_J(\mathfrak q(\bm q))=d(\mathfrak{c}) $.  Since the indices of $\alpha_1$ and $\alpha_2$ coincide, then $ ~^1\!\cS_{\mathfrak{c}}^{\mathfrak{q}(\bm{q})}\cong ~^2\!\cS_{\mathfrak{c}}^{\mathfrak{q}(\bm{q})}$ with $d_1(\mathfrak{q}(\bm{q}))=d_2(\mathfrak{q}(\bm{q}))$ for any $\mathfrak{q}(\bm{q})\in Q$. Thus, there exists a unitary operator $V \in \cA_{\mathfrak{c}}$ such that $V^\dagger O^1_{\mathfrak{c}} V =  O^2_{\mathfrak{c}} $ for any $O^1_{\mathfrak{c}}\in ~^1\!\cS_{\mathfrak{c}}^{\mathfrak{q}(\bm{q})} $, $ O^2_{\mathfrak{c}} \in ~^2\!\cS_{\mathfrak{c}}^{\mathfrak{q}(\bm{q})}$, and any $\mathfrak{q}(\bm{q})\in Q$. By translation invariance, we can reproduce the same argument and construct an analogous unitary map $V'$ on $\cA_{\mathfrak{q}(\bm{q})}$. Let $\{\Lambda_{\bm x}^1\}=\{\mathfrak{c}+2 \bm x\}_{\bm x \in \Z^s} $ and $\{\Lambda_{\bm x}^2\}=\{\mathfrak q(\bm q)+2\bm x\}_{\bm x \in \Z^s}$ be the Margolus partition of the lattice as per Definition~\ref{def:margolus}. We can take the unitary $V'$ to act in parallel on all the algebras defined by the elements of the partition $\{\Lambda_{\bm{x}}^2\}$, thus realizing a partitioned automorphism $\beta$ of the quasi-local algebra $\cA_{\Z^s}$. It follows that
    \begin{equation*}
        \cA_{\mathfrak c}\, {\overset{\alpha_1}{\cong}}\bigotimes_{\mathfrak q(\bm q) \in Q} ~^1\!\cS^{\mathfrak{c}}_{\mathfrak q(\bm q)} \overset{\beta}{\cong} \bigotimes_{\mathfrak q(\bm q) \in Q} ~^2\!\cS^{\mathfrak{c}}_{\mathfrak q(\bm q)}\, {\overset{\alpha_2^{-1}}{\cong}} \cA_{\mathfrak c}\,.
    \end{equation*}
    Therefore there exists a unitary operator $U \in \cA_{\mathfrak c}$ realizing the restriction of the automorphism $\alpha_2^{-1} \beta \alpha_1: \cA_{\Z^s} \to \cA_{\Z^s}$ to the algebra of the super-cell $\cA_{\mathfrak c}$. As done above, we can construct a partitioned automorphism $\gamma$ of $\cA_{\Z^s}$  by applying $U$ in parallel to the algebras identified by the partition $\{\Lambda_{\bm{x}}^1\}$. Ultimately, we obtaied a QCA $\gamma = \alpha_2^{-1} \beta \alpha_1 $, where $\gamma$ and $\beta$ are FDQCs. Hence
    \begin{equation}\label{marg_eqv}
        \alpha_2 = \mu\, \alpha_1
    \end{equation}
    with $\mu=\beta\, \widetilde{\gamma}^{-1} $ an FDQC, where $\widetilde{\gamma}^{-1}=\alpha_1\,\gamma^{-1}\,\alpha_1^{-1}$ is an FDQC possibly slightly deeper than $\gamma^{-1}$ (QCA form a group under composition, with FDQCs constituting a normal subgroup \cite{freedman2020classification}).

    According to  Lemma~\ref{lem:lemmas}, shift QCA $\tau^{\bm{x}} \in \mathscr{V}_p^s$ suffice to reproduce all the values of the index $\bm{\iota}$. From this observation and Eq.~(\ref{marg_eqv}) then it follows that for any $\alpha \in \mathscr{V}_p^s$ there exists a unique shift $\tau^{\bm{x}} \in \mathscr{V}_p^s$ (modulo FDQCs) such that 
    \begin{equation}\label{shift_equiv}
        \alpha = \mu\, \tau^{\bm{x}}\,,
    \end{equation}
with $\mu$ an FDQC. Therefore shifts constitute the representatives of equivalence classes of automata $\mathscr{V}_p^s$ modulo FDQCs. We establish the null shift (trivial evolution) $\tau^{\bm{0}}=\text{id}$ to represent the equivalence class of FDQCs in $ \mathscr{V}_p^s$ (no ancillary system is required). Indeed, $\bm \iota (\text{id})=\bm{1}$ implies that whether  $\bm{\iota}(\alpha)=\bm{1}$ then $\alpha$ is an MFDQC. Viceversa, since shifting an algebra along a given direction is essentially a one-dimensional shift automaton -- that we know is not an FDQC \cite{GNVW} -- if  $\alpha \in \mathscr{V}_p^s$ is an FDQC then, in compliance with Eq.~(\ref{shift_equiv}), it must be equivalent to a null shift, whose index is $\bm{1}$. 
\end{proof}
\noindent
The hypothesis of von Neumann neighborhood allows for precise determination of the dimensions $d(\mathfrak q(\bm q))$ of the support algebras on the quadrants (Figure~\ref{fig:dimsupp}(A)). Moreover, the hypothesis of cells of prime dimension avoids the possible ambiguity illustrated in Figure~\ref{fig:dimsupp}(B). For instance, considering two QCA $\alpha_1,\alpha_2 \in \mathscr{V}_{p^2}^2$ on $\Z^2$ with two qudits of prime dimension per site $\cA_{\bm{x}} \cong \cM_{p^2}$, they may give rise to non-isomorphic quadrants support algebras despite both having the same index $\bm{\iota}(\alpha_1)=\bm{\iota}(\alpha_2)=\bm{1}$. Although we are unable to explicitly construct an automaton $\alpha_2$ that realizes support algebras as in the right panel of Figure~\ref{fig:dimsupp}(B), we cannot exclude its existence. In other words, when the dimension of cells is nonprime, items~\textit{1.} and~\textit{3.} of Lemma~\ref{lem:lemmas} fail, potentially allowing for a broader range of dynamics.
\begin{figure}[h]
    \centering
    \includegraphics[width=8.0cm]{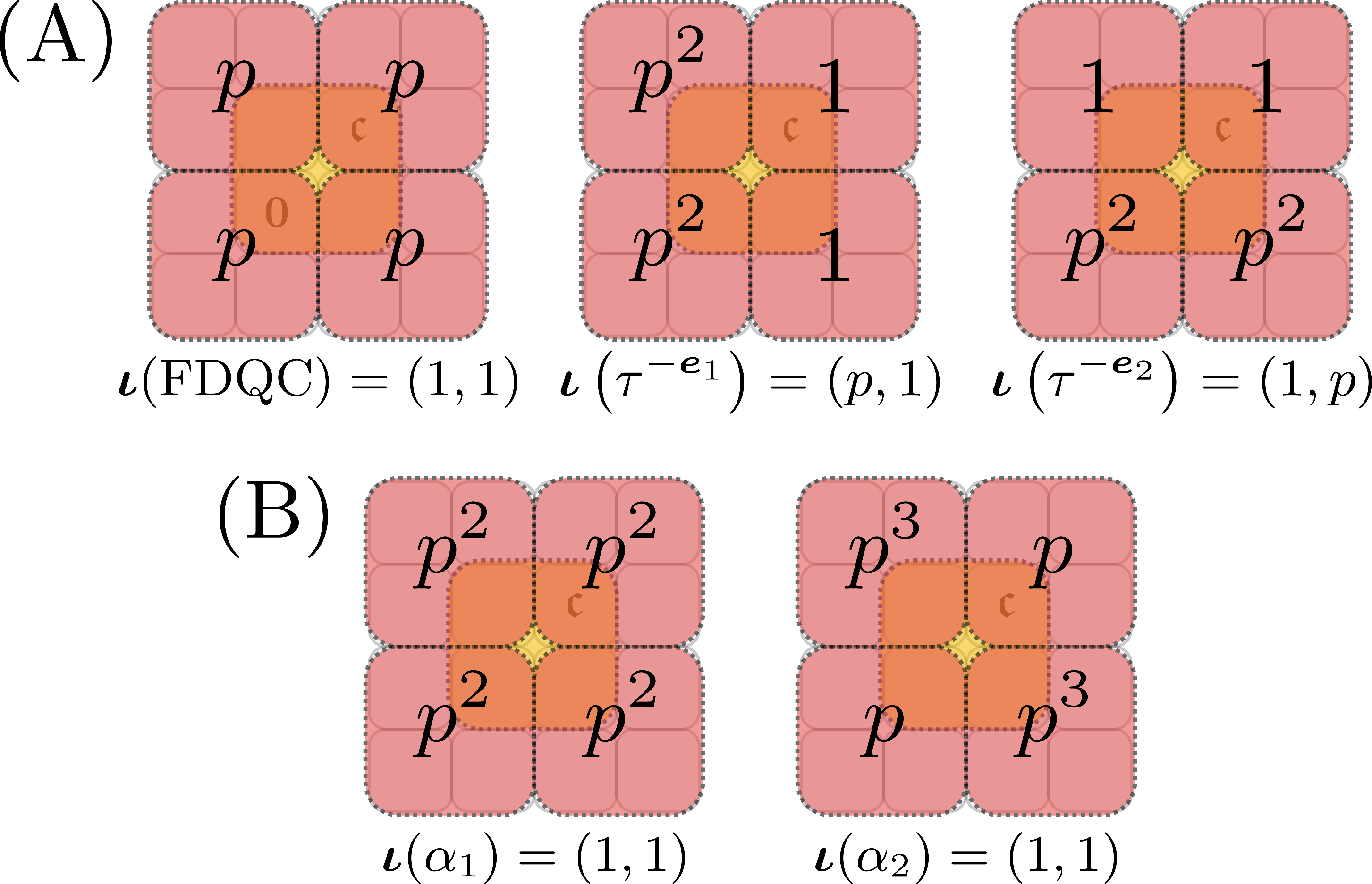}
    \caption{Dimensions of the support algebras of $\alpha(\cA_{\mathfrak{c}})$ on the quadrants $\mathfrak q(\bm q)$:  $\cS^{\mathfrak{c}}_{\mathfrak q(\bm q)} \cong \cM_{d(\mathfrak q(\bm q))}$ where $d(\mathfrak q(\bm q))=p^{n\left(\mathfrak q(\bm q)\right)}$. Panels (A) display the possibilities for QCA belonging to $\mathscr{V}_p^2$ and related indices. Panels (B) show an ambiguousness arising when considering $\mathscr{V}_d^s$ with $d$ nonprime.}
    \label{fig:dimsupp}
\end{figure}

\end{document}